\documentclass[format=acmsmall, review=false, screen=true]{acmart}

\usepackage{booktabs} 

\usepackage[ruled]{algorithm2e} 

\SetAlFnt{\small}
\SetAlCapFnt{\small}
\SetAlCapNameFnt{\small}
\SetAlCapHSkip{0pt}
\IncMargin{-\parindent}

\acmJournal{TALG}
\acmVolume{0}
\acmNumber{0}
\acmArticle{0}
\acmYear{2017}
\acmMonth{0}
\copyrightyear{2017}

\setcopyright{acmlicensed}

\acmDOI{0000001.0000001}

\received{February 2017}
\received[revised]{May 2018}
\received[accepted]{May 2018}

\usepackage{amsmath,amssymb,amsfonts,latexsym,epsfig,graphicx}
\usepackage{enumerate}
\usepackage{dsfont}
\usepackage{wrapfig}
\usepackage{hyperref}
\usepackage{float}

\usepackage{color}

\newtheorem{theorem}{Theorem}
\newtheorem{observation}[theorem]{Observation}

\long\def\ignore#1{}

\let\phi\varphi
\def\zet{{\mathbb Z}}

\let\subset\subseteq

\DeclareMathOperator\CSP{CSP}
\DeclareMathOperator\symdiff{\Delta}
\DeclareMathOperator\walk{walk}

\def\CSP{\operatorname{CSP}}
\DeclareMathOperator\Even{Even}
\DeclareMathOperator\Odd{Odd}

\newcommand\tinymod[1]{\operatorname{(mod\,#1)}}

\def\en{{\mathbb N}}
\def\CSP{\operatorname{CSP}}

\def\calC{{\mathcal C}}
\def\calE{{\mathcal E}}

\def\PlanarCSP{\operatorname{CSP}_{\textsc{planar}}}
\def\EdgeCSP{\operatorname{CSP}_{\textsc{edge}}}

\begin{document}
\title{Even Delta-Matroids and the Complexity of Planar Boolean CSPs
}

\author{Alexandr Kazda}
\orcid{0000-0002-7338-037X}
\affiliation{%
  \institution{IST Austria}
  \streetaddress{Am Campus 1}
  \city{Klosterneuburg}
  \postcode{3400}
  \country{Austria}}
\affiliation{%
  \institution{Charles University}
  \streetaddress{Sokolovsk\'a
83}
  \city{Prague}
  \postcode{186 75}
  \country{Czech Republic}}
\email{alex.kazda@gmail.com}
\author{Vladimir Kolmogorov}
\affiliation{%
  \institution{IST Austria}
  \streetaddress{Am Campus 1}
  \city{Klosterneuburg}
  \postcode{3400}
  \country{Austria}}
\email{vnk@ist.ac.at}
\author{Michal Rol\'inek}
\affiliation{%
  \institution{IST Austria}
  \streetaddress{Am Campus 1}
  \city{Klosterneuburg}
  \postcode{3400}
  \country{Austria}}
\affiliation{%
  \institution{Max Planck Institute for
  Intelligent Systems}
  \streetaddress{Max-Planck-Ring 4}
  \city{T\"ubingen}
  \postcode{72 076}
  \country{Germany}}
\email{michalrolinek@gmail.com}
\authorsaddresses{Alexandr Kazda (ORCID 0000-0002-7338-037X),
  Department of Algebra, Charles University, Sokolovsk\'a
83, 18675 Praha 8, Czech Republic; Vladimir Kolmogorov, IST Austria, Am Campus 1, 3400
  Klosterneuburg, Austria; Michal Rol\'inek, Max Planck Institute for
  Intelligent Systems, Max-Planck-Ring 4, 72076 T\"ubingen, Germany.}

\begin{abstract}
The main result of this paper is a generalization of the classical blossom
algorithm for finding perfect matchings. Our algorithm can efficiently solve
Boolean CSPs where each variable appears in
exactly two constraints (we call it edge CSP) and all constraints are \emph{even $\Delta$-matroid}
relations (represented by lists of tuples).
As a consequence of this, we settle the complexity classification
of planar Boolean CSPs started by Dvo\v r\'ak and Kupec.

Using a reduction to even $\Delta$-matroids, we then extend the tractability
result to larger classes of $\Delta$-matroids that we call {\em efficiently
coverable}. It properly includes classes that were known to be tractable
before, namely \emph{co-independent}, \emph{compact}, \emph{local},
\emph{linear} and \emph{binary}, with the following caveat: we represent
$\Delta$-matroids by lists of tuples, while the last two use a representation
by matrices. Since an $n\times n$ matrix can represent exponentially many
tuples, our tractability result is not strictly stronger than the known
algorithm for linear and binary $\Delta$-matroids. 
\end{abstract}

\ccsdesc[500]{Theory of computation~Design and analysis of algorithms}
\ccsdesc[100]{Mathematics of computing~Graph theory}

\keywords{Constraint satisfaction problem, delta-matroid, blossom
algorithm}

\maketitle

\section{Introduction}
The constraint satisfaction problem (CSP) has been a classical topic in
computer science for decades. Aside from its indisputable practical importance,
it has also heavily influenced theoretical research. The uncovered connections
between CSP and areas such as graph theory, logic, group theory, universal
algebra, or submodular functions provide some striking examples of the
interplay between CSP theory and practice.

We can exhibit such connections especially if we narrow our interest down to
\emph{fixed-template CSPs}, that is, to sets of constraint satisfaction
instances in which the constraints come from a fixed set of relations $\Gamma$.
For any fixed $\Gamma$ the set of instances $\CSP(\Gamma)$ forms a
decision problem; the question if $\CSP(\Gamma)$ is always either
polynomial-time solvable or NP-complete (in other words it avoids intermediate
complexities assuming $P \neq NP$) is known as the CSP dichotomy conjecture
\cite{feder-vardi}. After 20 years of effort by mathematicians and computer
scientists, the dichotomy conjecture seems to be finally
proved~\cite{bulatov-dichotomy}, \cite{zhuk-dichotomy}.

In this work we address two special structural restrictions for CSPs with
Boolean variables. One is limiting to at most two constraints per variable and
the other requires the constraint network to have a planar representation. The
first type, introduced by Feder \cite{feder-delta-matroids-fanout}, has very
natural interpretation as CSPs in which edges play the role of variables and
nodes the role of constraints, which is why we choose to refer to it as
\emph{edge CSP}. It was Feder who showed the following hardness result: Assume
the constraint language $\Gamma$ contains both unary constant relations (that
is, constant 0 and constant 1). Then unless 
all relations in $\Gamma$ are $\Delta$-matroids, the edge CSP with constraint
language $\Gamma$ has the same complexity as the unrestricted CSP with constraint
language $\Gamma$. Since then, there has also been progress on the algorithmic side.
Several tractable (in the sense of being polynomial time solvable) classes of
$\Delta$-matroids were identified \cite{feder-delta-matroids-fanout,feder-ford-matroids,Istrate97lookingfor,Dalmau2003,Geelen2003377,dvorak-kupec-planar-csp}. A
recurring theme is the connection between $\Delta$-matroids and matching
problems.

Recently, a setting for planar CSPs was formalized by Dvo\v{r}\' ak and Kupec
\cite{dvorak-kupec-planar-csp}. In their work, they provide certain hardness
results together with a reduction of the remaining cases to Boolean edge CSP.
Dvo\v r\'ak and Kupec's results imply that completing the complexity classification of Boolean planar
CSPs is equivalent to establishing the complexity of (planar) Boolean edge CSP where all
the constraints are \emph{even $\Delta$-matroids}. In their paper, Dvo\v{r}\'
ak and Kupec provided a tractable subclass of even $\Delta$-matroids along with
computer-aided evidence that the subclass (matching realizable even
$\Delta$-matroids) covers all even $\Delta$-matroids of arity at most 5.
However, it turns out that there exist even $\Delta$-matroids of arity 6 that
are not matching realizable; we provide an example of such a $\Delta$-matroid
in Appendix~\ref{sec:appendix}.

The main result of our paper is a generalization of the classical Edmonds'
blossom-shrinking algorithm for matchings~\cite{Edmonds:65} that we use to efficiently solve edge CSPs with even $\Delta$-matroid
constraints. This settles the complexity classification of planar CSP. Moreover, we give an extension of the algorithm to cover a wider class of
$\Delta$-matroids. This extension subsumes (to our best knowledge) all previously known tractable
classes. This paper is the journal version of the conference
article~\cite{kazda-kolmogorov-rolinek-soda}.

One notable problem that our paper leaves open is how to generalize the
Tutte-Berge formula\cite{tutte-matchings,berge-graph-theory}. By this formula,
the size of a maximum matching in a graph $G$ is $1/2 \cdot \min_{U\subset
V(G)}(|V(G)|+|U|-\mathrm{odd}(G-U))$ where $\mathrm{odd}(G-U)$ counts the
number of odd components of the graph we obtain from $G$ by removing all
vertices in $U$. Since edge CSPs generalize graph matchings, it would be
satisfying to obtain a similar formula for, say, the minimal number of
inconsistent variables in a solution of an edge CSP with even $\Delta$-matroid
constraints. However, we believe that obtaining such a formula will require
expanding our toolbox (particularly when it comes to situations involving
several blossoms at once), so we leave generalization of Tutte-Berge formula to
even $\Delta$-matroids as an open problem.

The paper is organized as follows. In the introductory Sections
\ref{section:prelim}, \ref{sec:implications}, and \ref{sec:matchings} we
formalize the frameworks, discuss how dichotomy for Boolean planar CSP follows from
our main theorem and sharpen our intuition by highlighting similarities
between edge-CSPs and perfect matching problems, respectively. The algorithm
is described in Section~\ref{section:algorithm} and the proofs required for
showing its correctness are in Section~\ref{sec:proofs}. The extension of the
algorithm is discussed in Section~\ref{sec:extend} and Appendix~\ref{app:covers}.

\section{Preliminaries}  \label{section:prelim}

\begin{definition} A Boolean CSP \emph{instance} $I$ is a pair
  $(V,\mathcal C)$ where $V$ is the set of \emph{variables} and $\mathcal C$
  the set of \emph{constraints} of $I$. A $k$-ary 
  constraint $C\in\mathcal C$ is a pair $(\sigma, R_C)$ where
  $\sigma\subset V$ is a set of size $k$ (called the \emph{scope} of $C$) and $R_C\subset \{0,1\}^\sigma$ is a relation on $\{0,1\}$. 
  A solution to $I$ is a mapping $\hat f:V\rightarrow\{0,1\}$ such that for
  every constraint $C=(\sigma,R_C)\in\calC$, $\hat f$ restricted to $\sigma$ lies in $R_C$.
\end{definition}

\begin{definition}
If all constraint relations of $I$ come from a set of relations $\Gamma$
(called the \emph{constraint language}), we say that $I$ is a
$\Gamma$-instance. For $\Gamma$ fixed, we will denote the problem of deciding if a
$\Gamma$-instance given on input has a solution by $\CSP(\Gamma)$.
\end{definition}

Note that the above definition is not fully general in the sense that it
\emph{does not allow one variable to occur multiple times in a constraint}; we
have chosen to define Boolean CSP in this way to make our notation a bit
simpler.  This can be done without loss of generality as long as $\Gamma$
contains the equality constraint (i.e.~$\{(0,0),(1,1)\}$): If a variable, say
$v$, occurs in a constraint multiple times, we can add extra copies of $v$ to
our instance and join them together by the equality constraint to obtain a
slightly larger instance that satisfies our definition. 

For brevity of notation, we will often not distinguish a constraint $C\in\calC$ from its constraint relation $R_C$;
the exact meaning of $C$ will always be clear from the context.
Even though in principle different constraints can have the same constraint
relation, our notation would get cumbersome if we wrote $R_C$ everywhere.

The main point of interest is classifying the computational complexity of
$\CSP(\Gamma)$. Constraints of an instance are
specified by lists of tuples in the corresponding relations and thus those
lists are considered to be part of the input. We will say that $\Gamma$ 
contains the unary constant relations if $\{(0)\},\{(1)\}\in \Gamma$ (these relations allow us to
fix the value of a certain variable to 0 or 1).

For Boolean CSPs (where variables are assigned Boolean values), the complexity
classification of $\CSP(\Gamma)$ due to Schaefer has been known for a long
time~\cite{Schaefer78:complexity}. There was much progress since then,
including a full classification for the three-element
domain~\cite{bulatov-2006-three-csp} and for conservative
structures~\cite{Bulatov03:conservative}. Recently, two proofs of
classification in the general case were presented at the FOCS conference~\cite{bulatov-dichotomy},
\cite{zhuk-dichotomy}. However, in this work we concentrate on Boolean domains only.

Our main focus is on restricted forms of the CSP. In particular, we are
interested in structural restriction, i.e.~in restriction on the constraint
network. Once one starts to limit the
shape of instances, the Boolean case becomes complicated again. (As a side note,
we expect 
similar problems for larger domains to be very hard to classify. For example,
Dvo\v r\'ak and Kupec note that one can encode coloring planar graphs by four colors as a
class of planar CSPs that always have a solution for a highly nontrivial
reason, namely the four color theorem.)

A natural structural restriction would be to limit the number of constraints
in whose scope a variable can lie. When $k\geq 3$ and $\Gamma$ contains all unary constants, then $\CSP(\Gamma)$ with each variable in at most $k$ constraints is
polynomial time equivalent to unrestricted $\CSP(\Gamma)$, see~\cite[Theorem~2.3]{Dalmau2003}. This leaves instances with at most two occurrences per variable in the
spotlight. To make our arguments clearer, we will assume that each variable occurs
exactly in two constraints 
(following~\cite{feder-delta-matroids-fanout}, we can reduce decision CSP
instances with at most two appearances of each variable to instances with
exactly two appearances by taking two copies of the instance and identifying
both copies of $v$ whenever $v$ is a variable that originally appeared in only
a single constraint).

\begin{definition}[Edge CSP] Let $\Gamma$ be a constraint language. Then the
  problem $\EdgeCSP(\Gamma)$ is the restriction of $\CSP(\Gamma)$ to those instances 
  in which every variable is present in exactly two constraints.
\end{definition}

Perhaps a more natural way to look at an instance $I$ of an edge CSP is to
consider a graph whose edges correspond to variables of $I$ and nodes to
constraints of $I$. Constraints (nodes) are incident with variables (edges)
they interact with.  In this (multi)graph, we are looking for a satisfying
Boolean edge labeling.  Viewed like this, edge CSP becomes a counterpart to the
usual CSP where variables are typically identified with nodes and
constraints with (hyper)edges. The idea of ``switching'' the role of
(hyper)edges and vertices already appeared in the counting CSP community under
the name Holant problems \cite{Holant-problem-journal-version}.

This type of CSP is sometimes called ``binary CSP'' in the literature
\cite{dvorak-kupec-planar-csp}. However, this term is very commonly used for
CSPs whose all constraints have arity at most two \cite{books/daglib/0076790}.
In order to resolve this confusion (and for the reasons described in the
previous paragraph), we propose the term ``edge CSP''.

As we said above, we will only consider Boolean edge CSP, often omitting the word
``Boolean'' for space reasons. The following Boolean-specific definitions
will be useful to us:

\begin{definition}
  Let $f\colon V\to\{0,1\}$ (we will denote the set of all such
  mappings $f$ by $\{0,1\}^V$) and let $v\in V$. We will denote by $f\oplus v$ the
  mapping $V\to\{0,1\}$ that agrees with $f$ on $V\setminus\{v\}$ and has value
  $1-f(v)$ on $v$. For a set $S=\{s_1,\dots,  s_k\} \subset V$ we let $f \oplus S = f \oplus s_1 \oplus \dots \oplus s_k$. Also for $f,g \colon V\to\{0,1\}$ let $f \symdiff g \subset V$ be the set of variables $v$ for which $f(v) \neq g(v)$.
\end{definition}

\begin{definition}
  Let $V$ be a set. A nonempty subset $M$ of
  $\{0,1\}^V$ is called a \emph{$\Delta$-matroid} if whenever $f,g \in M$ and $v \in f \symdiff g$, then there exists $u \in f \symdiff  g$ such that
  $f \oplus \{u, v\} \in M$.
If moreover, the parity of the number of ones over all tuples of $M$ is
constant, we have an \emph{even $\Delta$-matroid} (note that in that case we
never have $u = v$ so $f \oplus \{u, v\}$ reduces to $f \oplus u \oplus v$).
\end{definition}

The \emph{$\Delta$-matroid parity} problem~\cite[Problem (23)]{jensen-korte} has as its input a
$\Delta$-matroid $M\subset \{0,1\}^E$ and a partition $P$ of $E$ into pairs.
The goal is to find $\alpha\in M$ such that $\alpha(u)\neq\alpha(v)$ for as few
pairs $\{u,v\}\in P$ as possible. This problem is easily equivalent to finding
an edge labeling that minimizes the number of inconsistent edges of the edge
CSP with edges (variables) $E$, binary equality constraints on all pairs in $P$
and one big constraint $M$ with the scope $E$ (see
Definition~\ref{def:Solutions} for an exact definition of what we mean by an
edge labeling and inconsistent edges).

A $\Delta$-matroid with all tuples containing exactly the same number of ones
is (the set of bases of) a matroid. There is a vast body of literature on the
properties of matroids; here we only mention two notions that are immediately
relevant to edge CSP: The \emph{matroid parity} problem is the $\Delta$-matroid parity
problem where $M$ is a matroid. In the literature, the matroid parity problem
is usually formulated in the equivalent way ``find $\alpha\in M$ so that
$\alpha(u)=\alpha(v)=1$ for as many pairs $\{u,v\}\in P$ as possible.''

A similar problem is the \emph{matroid matching} problem where we are given a
graph $G$ (with vertex set $V(G)$ and edge set $E(G)$) and a matroid $M$ on the variable set $V(G)$ and are looking for
$\alpha\in M$ such that the subgraph of $G$ induced by $\{v\in V(G)\colon
\alpha(v)=1\}$ contains as big a matching as possible. It is straightforward
to verify that this problem is equivalent to finding an edge labeling that minimizes the number
of inconsistent edges in the edge CSP instance with variable set
$V(G)\cup E(G)$, one big constraint $M$ on $V(G)$ and a constraint $C_v$ for
each $v\in V(G)$. The scope of $C_v$ consists of $v$ and all $e\in E(G)$
incident with $v$ in $G$. The relation of $C_v$ contains the all zero tuple
$(0,\dots,0)$ and the tuple $(0,\dots,0)\oplus v\oplus e$ for each $e\in E(G)$ incident
with $v$.

We note for future reference that (even) $\Delta$-matroids are
closed under gadget constructions, known as compositions in
$\Delta$-matroid theory: If $M\subset \{0,1\}^U$ and $N\subset
\{0,1\}^V$ are $\Delta$-matroids defined on two sets of variables
such that the set symmetric difference of $U$ and $V$, denoted by
$U\symdiff V$, is nonempty, we define the composition of $M$ and
$N$ to be the relation
\begin{align*}
	\{\gamma\in \{0,1\}^{U\symdiff V}\colon \exists \alpha \in M,\exists \beta\in
	N,\,\forall &u\in U\cap V,\alpha(u)=\beta(u),\\
	\forall &u\in U\setminus V,\,\gamma(u)=\alpha(u),\\
	\forall &v\in V\setminus U,\,\gamma(v)=\beta(v)\}.
\end{align*}
\begin{proposition}[\cite{bouchet-cunningham-1995}]\label{prop:composition}
	The composition of two $\Delta$-matroids is a $\Delta$-matroid.
\end{proposition}
Moreover, a quick parity argument gives us that the composition of
two even $\Delta$-matroids must be an even $\Delta$-matroid.

The strongest hardness result on Boolean edge CSP is from Feder.

\begin{theorem}[\cite{feder-delta-matroids-fanout}] If $\Gamma$ is a constraint language containing unary constant relations such that $\CSP(\Gamma)$ is NP-Hard and there is $R \in \Gamma$ which is not a $\Delta$-matroid, then $\EdgeCSP(\Gamma)$ is NP-Hard.
\end{theorem}

Tractability was shown for special classes of $\Delta$-matroids, namely 
binary~\cite{Geelen2003377,Dalmau2003},
linear~\cite{Geelen2003377}\footnote{The paper~\cite{Geelen2003377} actually
showed
tractability of the $\Delta$-matroid parity problem with linear or binary
constraints. However, given representations of constraints of an edge CSP by
matrices $M_1,M_2,\dots$, like in~\cite{Geelen2003377}, a block matrix with
$M_1,M_2,\dots$ on the diagonal and zeroes elsewhere represents a ``big''
$\Delta$-matroid parity problem (with a suitably chosen pairing) which, when
solved, gives a solution of the original edge CSP.},
co-independent~\cite{feder-delta-matroids-fanout}, compact~\cite{Istrate97lookingfor}, and
local~\cite{Dalmau2003} (see the definitions in the respective papers). 
All the proposed algorithms are based on variants of searching for
augmenting paths. In this work we propose a more general algorithm that
involves both augmentations and contractions. In particular, we prove the
following.

\begin{theorem}\label{thm:main} If $\Gamma$ contains only even $\Delta$-matroid relations, then $\EdgeCSP(\Gamma)$ can be solved in polynomial time.
\end{theorem}

Our algorithm will in fact be able to solve even a certain optimization version of the edge CSP (corresponding to finding a maximum matching). This is discussed in detail in Section \ref{section:algorithm}.

In Section~\ref{sec:extend} we show that if a class of $\Delta$-matroids is
\emph{efficiently coverable}, then it defines a tractable CSP. The whole
construction is similar to, but more general than, $\mathcal C$-zebra $\Delta$-matroids
introduced in~\cite{feder-ford-matroids}. We note here also that the class of
coverable $\Delta$-matroids is natural in the sense of being
closed under gadget constructions (also known as composition of
$\Delta$-matroids) which we split into taking direct products and
identifying variables.

\begin{definition}\label{def:even-neighbors}
  Let $M$ be a $\Delta$-matroid. We say that $\alpha,\beta\in M$ are
  \emph{even-neighbors} if there exist distinct variables $u,v\in V$ such that
  $\beta=\alpha\oplus u\oplus v$ and $\alpha\oplus u\not\in M$. We say we can
  \emph{reach} $\gamma\in M$ from $\alpha\in M$ if there is a chain
  $\alpha=\beta_0,\beta_1,\dots,\beta_n=\gamma$ where each pair
  $\beta_i,\beta_{i+1}$ are even-neighbors.
\end{definition}
\begin{definition}\label{def:coverable}
  We say that $M$ is \emph{coverable} if for every $\alpha\in M$
  there exists $M_\alpha$ such that:
  \begin{enumerate}
    \item $M_\alpha$ is an even $\Delta$-matroid (over the same ground set as $M$),
    \item $M_\alpha$ contains all $\beta\in M$ that can be reached from
      $\alpha$ (including $\alpha$ itself),
    \item whenever $\gamma\in M$ can be reached
      from $\alpha$ and $\gamma\oplus u\oplus v\in M_\alpha\setminus M$, 
      then $\gamma\oplus u,\gamma\oplus v\in M$.
  \end{enumerate}
\end{definition}

In our algorithm, we will need to have access to the sets $M_\alpha$, so we
need to assume that all our $\Delta$-matroids, in addition to being coverable,
come from a class of $\Delta$-matroids where the sets $M_\alpha$ can be
determined quickly. This is what \emph{efficiently coverable} means (for a
formal definition see Definition~\ref{def:efficiently-coverable}).

The following theorem is a strengthening of a result
from~\cite{feder-ford-matroids}:

\begin{theorem}\label{thm:extension2} Given an edge CSP instance $I$ with efficiently coverable $\Delta$-matroid constraints,
  an optimal edge labeling (i.e. edge labeling having fewest possible
  inconsistently labeled edges) $f$ of $I$ can be found in time polynomial in $|I|$.
In particular, $\EdgeCSP(\Gamma)$ can be solved in polynomial time.
\end{theorem}

As we show in Appendix~\ref{app:covers}, efficiently coverable $\Delta$-matroid
classes include numerous known tractable classes of $\Delta$-matroids: $\mathcal
C$-zebra $\Delta$-matroids~\cite{feder-ford-matroids} for any subclass $\calC$
of even $\Delta$-matroids (where we assume, just like
in~\cite{feder-ford-matroids}, that we are given the zebra representations on
input) as well as co-independent~\cite{feder-delta-matroids-fanout}, compact
\cite{Istrate97lookingfor}, local~\cite{Dalmau2003},
linear~\cite{Geelen2003377} and binary~\cite{Geelen2003377,Dalmau2003}
$\Delta$-matroids.  To our best knowledge these are all the known tractable
classes and according to \cite{Dalmau2003} the classes other than
$\calC$-zebras are pairwise incomparable.  

One caveat of our result when applied to linear or binary $\Delta$-matroids,
which does not allow us to say that our algorithm generalizes everything that
came before, is that our representation of $\Delta$-matroids (by lists of
tuples) is different from e.g.~\cite{Geelen2003377} where linear and binary
$\Delta$-matroids are represented by matrices. A linear $\Delta$-matroid
described by an $n\times n$ matrix can contain exponential number of tuples,
making our algorithm inefficient when constraints are encoded by matrices on
the input.

\section{Implications} \label{sec:implications}

In this section we explain how our result implies full complexity classification of planar Boolean~CSPs.

\begin{definition} Let $\Gamma$ be a constraint language. Then
  $\PlanarCSP(\Gamma)$ is the restriction of $\CSP(\Gamma)$ to the set of
  instances for which there exists a planar graph $G(V, E)$ such that $v_1$, \dots, $v_k$ is a face of $G$ (with nodes listed in counter-clockwise order) if and only if there is a unique constraint imposed on the tuple of variables $(v_1, \dots, v_k)$.
\end{definition}

It is also noted in \cite{dvorak-kupec-planar-csp} that checking whether an
instance has a planar representation can be done efficiently (see
e.g..~\cite{hopcroft-tarjan-planarity}) and hence it does not matter if we are given a planar drawing of $G$ as a part of the input or not. The planar restriction does lead to new tractable cases, for example planar {\sc NAE-3-SAT} (Not-All-Equal 3-Satisfiability) \cite{Moret:1988:PNP:49097.49099}.

\begin{definition} A relation $R$ is called \emph{self-com\-ple\-mentary} if for all
  $T\in\{0,1\}^n$ we have $T \in R$ if and only if $T \oplus \{1,2,\dots,n\}\in
  R$ (i.e. $R$ is invariant under simultaneous flipping of all entries of a tuple).
\end{definition}

\begin{definition}
  For a tuple of Boolean variables $T = (t_1, \dots, t_n)\in \{0,1\}^n$, let
  \[
    dT = \{ t_i+ t_{i+1}\tinymod2\colon i=1,2,\dots,n\}
  \]
  (we take $t_{n+1}=t_1$ here). For a relation $R$ and a set of relations $\Gamma$, let 
  $dR = \{dT \colon T \in R\}$ and $d\Gamma = \{ dR \colon R \in \Gamma\}$.
\end{definition}

Since self-complementary relations don't change when we flip all their
coordinates, we can describe a self-complementary relation by looking at the
differences of neighboring coordinates; this is exactly the meaning of $dR$.
Note that these differences are realized over edges of the given planar graph.

Knowing this, it is not so difficult to imagine that via switching to the planar dual of
$G$, one can reduce a planar CSP instance to some sort of edge CSP instance. This is in fact
part of the following theorem from~\cite{dvorak-kupec-planar-csp}:

\begin{theorem}\label{thm:dvorakkupec} Let $\Gamma$ be such that $\CSP(\Gamma)$
  is NP-Hard. Then:
  \begin{enumerate}
\item If there is $R \in \Gamma$ that is not self-complementary, then $\PlanarCSP(\Gamma)$ is NP-Hard.
\item If every $R \in \Gamma$ is self-complementary and there exists $R \in \Gamma$ such that $dR$ is not even $\Delta$-matroid, then $\PlanarCSP(\Gamma)$ is NP-Hard.
\item If every $R \in \Gamma$ is self-complementary and $dR$ is an even $\Delta$-matroid, then $\PlanarCSP(\Gamma)$ is polynomial-time reducible to $$\EdgeCSP(d\Gamma \cup \{EVEN_1, EVEN_2, EVEN_3\})$$ where
$EVEN_i = \{(x_1, \dots x_i) \colon x_1 + \cdots + x_i \equiv 0 \pmod 2\}.$ 
\end{enumerate}
\end{theorem}

Using Theorem~\ref{thm:main}, we can finish this classification:

\begin{theorem}[Dichotomy for $\PlanarCSP$] Let $\Gamma$ be a constraint
  language. Then $\PlanarCSP(\Gamma)$ is solvable in polynomial time if either
  \begin{enumerate}
\item $\CSP(\Gamma)$ is solvable in polynomial time or;
\item $\Gamma$ contains only self-complementary relations $R$ such that $dR$ is an even $\Delta$-matroid.
\end{enumerate}
Otherwise, $\PlanarCSP(\Gamma)$ is NP-Hard.
\end{theorem}
\begin{proof} By Theorem \ref{thm:dvorakkupec} the only unresolved case reduces
  to solving $$\EdgeCSP(d\Gamma \cup \{EVEN_1, EVEN_2, EVEN_3\}).$$ Since the
  relations $EVEN_i$ are even $\Delta$-matroids for every $i$, this is
  polynomial-time solvable thanks to Theorem \ref{thm:main}.
\end{proof}

\section{Even $\Delta$-matroids and matchings} \label{sec:matchings}

In this section we highlight the similarities and dissimilarities between even
$\Delta$-matroid CSPs and matching problems. These similarities will guide us
on our way through the rest of the paper.

\begin{example} \label{ex:instance} For $n \in \mathbb{N}$ consider the
  ``perfect matching'' relation $M_n \subset \{0,1\}^n$ containing precisely
  the tuples in which exactly one coordinate is set to one and all others to
  zero. Note that $M_n$ is an even $\Delta$-matroid for all $n$. Then the instance $I$ of $\EdgeCSP(\{ M_n \colon n \in \mathbb{N} \})$ (represented in Figure \ref{pic:example}) is equivalent to deciding whether the graph of the instance has a perfect matching (every node is adjacent to precisely one edge with label 1).

One may also construct an equivalent instance $I'$ by ``merging''
 some parts of
the graph (in Figure \ref{pic:example} those are $X$ and $Y$) to single constraint nodes
(this is exactly composition of $\Delta$-matroids). The constraint relations imposed on the ``supernodes'' record sets of outgoing
edges which can be extended to a perfect matching on the subgraph induced by
the ``supernode''. For example, in the instance $I'$ the constraints imposed on
$X$ and $Y$ would be (with variables ordered as in Figure \ref{pic:example}):
\begin{align*}
X &= \{10000, 01000, 00100, 00010, 11001, 10101, 10011 \},\\
Y &= \{001, 010, 100, 111 \}.
\end{align*}
It is easy to check that both $X$ and $Y$ are even $\Delta$-matroids.
\end{example}

\begin{figure*}[t]
\begin{center}
\includegraphics[scale=1]{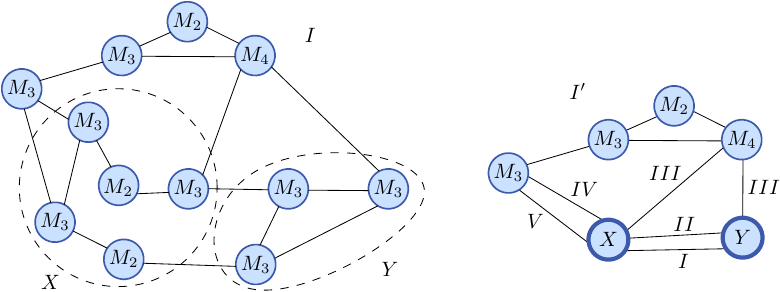}
\caption{On the left we see an instance $I$ that is equivalent to testing for perfect matching of the given graph. On the right is an equivalent instance $I'$ with contracted ``supernodes'' $X$ and $Y$.}
\label{pic:example}
\end{center}
\end{figure*}

One takeaway from this example is that any algorithm that solves edge CSP for the
even $\Delta$-matroid case has to work for perfect matchings in graphs as well. Another is the construction of even $\Delta$-matroids $X$ and $Y$ which can be generalized as follows.

\begin{definition}[Matching realizable relations] Let $G$ be a graph and let
  $v_1, \dots, v_a \in V(G)$ be distinct nodes of $G$. For an $a$-tuple $T =
  (x_1, \dots, x_a) \in \{0,1\}^a$, we denote by $G_T$ the graph obtained from
  $G$ by deleting all nodes $v_i$ such that $x_i=1$. Then we can define
  \[
    M(G, v_1, \dots, v_a) = \{T \in \{0,1\}^a \colon G_T \, \text{has a perfect matching} \}.
  \]
  We say that a relation $R \in \{0,1\}^a$ is \emph{matching realizable} if $R = M(G, v_1, \dots, v_a)$ for some graph $G$ and nodes $v_1, \dots, v_a \in V(G)$.
\end{definition}

Every matching realizable relation is an even $\Delta$-matroid
\cite{bouchet-matchings}.  Also, it should be clear from the definition
and the preceding example that $\EdgeCSP(\Gamma)$ is tractable if $\Gamma$
contains only matching realizable relations (assuming we know the graph $G$ and
the nodes $v_1,\dots,v_a$ for each relation): One can simply replace each constraint node with the
corresponding graph and then test for existence of perfect matching.

The authors of \cite{dvorak-kupec-planar-csp} also verify that every even $\Delta$-matroid of arity at most 5 is matching realizable. 
However, as we prove in Appendix~\ref{sec:appendix}, this is not true for
higher arities.

\begin{proposition}\label{prop:notmatchrel} There exists an even $\Delta$-matroid of arity 6 which is not matching realizable. 
\end{proposition}

Proposition \ref{prop:notmatchrel} shows that we cannot hope to simply
replace the constraint nodes by graphs and run the Edmonds' algorithm. The
$\Delta$-matroid constraints can exhibit new and more complicated behavior than
just matchings in graphs, as we shall soon see.  In fact, there is a known
exponential lower bound for the matroid parity problem (matroids
being special cases of even $\Delta$-matroids and matroid
parity being a special case of edge CSP, see above) where $M$ is given by
an oracle (i.e. not explicit lists of tuples)~\cite{jensen-korte} (see also a
related result by L. Lov\'asz~\cite{lovasz-matroids} that considers a problem
slightly different from matroid matchings), which rules out any
polynomial time algorithm that would work in the oracle model. In particular,
we are convinced that our method of contracting blossoms cannot be
significantly simplified while still staying polynomial time computable.

\section{Algorithm} \label{section:algorithm}

\subsection{Setup}
\label{section:setup}

We can draw edge CSP instances as constraint graphs: The \emph{constraint
graph} $G_I=(V\cup \calC,\calE)$ of $I$ is a bipartite graph with parts
$V$ and $\mathcal C$.  There is an edge $\{v,C\}\in \calE$ if and only if $v$
belongs to the scope of $C$. Throughout the rest of the paper we use lower-case
letters $u,v,x,y,\ldots$ for variable nodes in $V$ and upper-case letters $A,B,C,\ldots$ for constraint nodes in $\calC$. 
Since we are dealing with edge CSP, the degree of each node $v\in V$ in $G_I$
is exactly two and since we don't allow a variable to appear in a constraint
twice, $G_I$ has no multiple edges.
For such instances $I$ we introduce the following terminology and notation.

\begin{definition}\label{def:Solutions} An \emph{edge labeling} of $I$ is a mapping $f:\calE\rightarrow\{0,1\}$.
  For a constraint $C\in\calC$ with the scope $\sigma$ we will denote by $f(C)$ the tuple
  in $\{0,1\}^\sigma$ such that $f(C)(v)=f(\{v,C\})$ for all $v\in\sigma$. Edge labeling $f$
  will be called \emph{valid} if $f(C)\in C$ for all $C\in\calC$.

  Variable $v\in V$ is called \emph{consistent} in $f$ if
  $f(\{v,A\})=f(\{v,B\})$ for the two distinct edges $\{v,A\},\{v,B\}\in\calE$ of $G_I$.
  Otherwise, $v$ is \emph{inconsistent} in $f$.

  A valid edge labeling $f$ is \emph{optimal} if its number of inconsistent
  variables is minimal among all valid edge labelings of $I$. 
  Otherwise $f$ is called \emph{non-optimal}.
\end{definition}
Note that $I$ has a solution if and only if an optimal edge labeling $f$ of $I$ has no inconsistent variables.


The main theorem we prove is the following strengthening of Theorem \ref{thm:main}.

\begin{theorem}
Given an edge CSP instance $I$ with even $\Delta$-matroid constraints,
an optimal edge labeling $f$ of $I$ can be found in time polynomial in $|I|$.
\end{theorem}

\paragraph{Walks and blossoms}

When studying matchings in a graph, paths and augmenting paths are important. We will use
analogous objects, called $f$-walks and augmenting $f$-walks, respectively.

\begin{definition}\label{def:walk}
  A \emph{walk} $q$ of length $k$ in the instance $I$ is a sequence
$q_0 C_1 q_1 C_2 \dots C_{k} q_k$ where the variables $q_{i-1},q_i$ lie in the
scope of the constraint $C_i$, and each edge $\{v,C\}\in\calE$ is traversed at most once:
 $vC$ and $Cv$ occur in $q$ at most once, and they do not occur simultaneously. 

We allow walks of length 0 (i.e. single vertex walks) for formal reasons.
\end{definition}

Note that a walk in the instance $I$ can be viewed as a walk in the graph $G_I$
that starts and ends at nodes in $V$ and uses each edge at most once. Since
each node $v\in V$ has degree two in $G_I$, a walk that enters a variable node
$v$ through an edge must leave $v$ through the other edge and cannot ever return
to $v$ again. The two exceptional vertices are the initial and terminal vertex
of a walk. These vertices \emph{can} be identical, i.e. we allow walks of the form
$vCq_1 \ldots q_{k-1}Dv$. 

A \emph{subwalk} of $q$, denoted by $q_{[i,j]}$, is the walk $q_i C_{i+1} \dots
C_j q_j$ (again, we need to start and end in a variable).  The inverse walk to
$q$, denoted by $q^{-1}$, is the sequence $q_k C_k \dots q_1 C_1 q_0$. Given
two walks $p$ and $q$ such that the last node of $p$ is the first node of $q$,
we define their concatenation $pq$ in the natural way.  If
$p=\alpha_1\ldots\alpha_k$ and $q=\beta_1\ldots\beta_\ell$ are sequences of
nodes of a graph where $\alpha_k$ and $\beta_1$ are different but adjacent, we
will denote the sequence $\alpha_1\ldots\alpha_k\beta_1\ldots\beta_\ell$ also
by $pq$ (or sometimes as $p,q$).

If $f$ is an edge labeling of $I$ and $q$ a walk in $I$, we denote by $f\oplus
q$ the mapping that takes $f$ and flips the values on all variable-constraint
edges encountered in $q$, i.e.
\begin{equation}
  (f\oplus q)(\{v,C\})=\!\begin{cases}
    \!1\!-\!f(\{v,C\}) &\!\!\mbox{if $q$ contains $vC$ or $Cv$} \\
    \!f(\{v,C\}) &\!\!\mbox{otherwise.}    
  \end{cases}
\label{eq:foplusq:def}
\end{equation}

\begin{definition}\label{def:f-walk}
  Let $f$ be a valid edge labeling of an instance $I$.  
  A walk $q=q_0 C_1 q_1 C_2 \dots C_{k} q_k$ 
  with $q_0\ne q_k$ will be called an \emph{$f$-walk}
  if
  \begin{enumerate}
    \item variables $q_1,\ldots,q_{k-1}$ are consistent in $f$, and 
    \item\label{itm:valid-interval} $f\oplus q_{[0,i]}$ is a valid edge labeling for any $i\in[1,k]$.
  \end{enumerate}
  If in addition variables $q_0$ and $q_k$ are inconsistent in $f$ then $q$ will be called an \emph{augmenting $f$-walk}.
\end{definition}
Observe that condition~\ref{itm:valid-interval} of the definition of an
$f$-walk is stronger than just ``$f\oplus q$ is valid.'' Instead, an $f$-walk 
corresponds to a whole sequence of valid labelings.

Later we will show that a valid edge labeling $f$ is non-optimal
if and only if there exists an augmenting $f$-walk.
Note that one direction is straightforward: 
If $p$ is an augmenting $f$-walk, then $f\oplus p$ is valid and has 2 fewer inconsistent variables than $f$. 


Another structure used by the Edmonds' algorithm for matchings is a
\emph{blossom}. The precise definition of a blossom in our setting
(Definition~\ref{def:Blossom}) is a bit technical. Informally, an $f$-blossom is a walk 
  $b=b_0 C_1 b_1 C_2 \dots C_{k} b_k$ with $b_0=b_k$ such that:
  \begin{enumerate}
    \item variable $b_0=b_k$ is inconsistent in $f$ while variables $b_1,\ldots,b_{k-1}$ are consistent, and
    \item $f\oplus b_{[i,j]}$ is a valid
  edge labeling for any non-empty proper subinterval $[i,j]\subsetneq[0,k]$,
    \item there are no bad shortcuts inside $b$ (we will make this precise
      later).
  \end{enumerate}

\subsection{Algorithm description}
\label{section:algorithm-description}
We are given an instance $I$ of edge CSP with even $\Delta$-matroid constraints
together with a valid edge labeling $f$ and we want to either show that $f$ is optimal or
improve it.
Our algorithm will explore the graph $(V\cup \calC,\calE)$ building a directed forest $T$.
Each variable node $v\in V$ will be added to $T$ at most once.
Constraint nodes $C\in \calC$, however, can be added to $T$ multiple times. To
tell the copies of $C$ apart (and to keep track of the order in which we built
$T$), we will mark each $C$ with a timestamp $t\in
\en$; the resulting node of $T$ will be denoted as $C^t\in\calC\times\en$.
Thus, the forest will have the form $T=(V(T)\cup \calC(T),E(T))$
where $V(T)\subseteq V$ and $\calC(T)\subseteq\calC\times\en$.

The roots of the forest $T$ will be the inconsistent nodes of the instance
(for current $f$); all non-root nodes in $V(T)$ will be consistent.
The edges of $T$ will be oriented towards the leaves.
Thus, each non-root node $\alpha\in V(T)\cup \calC(T)$ will have
exactly one parent $\beta\in V(T)\cup \calC(T)$ with $\beta\alpha\in E(T)$. 
For a node $\alpha\in V(T)\cup \calC(T)$ let $\walk(\alpha)$ be the the unique path in $T$ from a root to $\alpha$.
Note that $\walk(\alpha)$ is a subgraph of $T$. 
Sometimes we will treat walks in $T$ as sequences of nodes in $V\cup\calC$
discussed in Section~\ref{section:setup}
(i.e.~with timestamps removed); such places should be clear from the context.

We will grow the forest $T$ in a greedy manner as shown in
Algorithm~\ref{alg:Improve}.  The structure of the algorithm resembles that of
the Edmonds' algorithm for matchings~\cite{Edmonds:65}, with the following
important distinctions: First, in the Edmonds' algorithm each ``constraint
node'' (i.e.~each node of the input graph) can be added to the forest at most
once, while in Algorithm~\ref{alg:Improve} some constraints $C\in\calC$ can be
added to $T$ and ``expanded'' multiple times (i.e.~$E(T)$ may contain edges
$C^su$ and $C^{t}w$ added at distinct timestamps $s\ne t$).  This is because we
allow more general constraints.  In particular, if $C$ is a ``perfect
matching'' constraint
(i.e.~$C=\{(a_1,\ldots,a_k)\in\{0,1\}^k\::\:a_1+\ldots+a_k=1\}$) then
Algorithm~\ref{alg:Improve} will expand it at most once. (We will not use this
fact, and thus omit the proof.)

Note that even when we enter a constraint node for the second or third time, we
``branch out'' based on transitions $vCw$ available before the first visit,
even though the tuple of $C$ might have changed in the meantime.
This could cause one to doubt that Algorithm~\ref{alg:Improve} works at all.
\begin{algorithm*}[t]
  {\bf Input:} Instance $I$, valid edge labeling $f$ of $I$.

  \noindent {\bf Output:} A valid edge labeling $g$ of $I$ with fewer inconsistent variables than
  $f$, or ``No'' if no such $g$ exists.

  \begin{enumerate}
    \item Initialize $T$ as follows: set timestamp $t=1$, and for each inconsistent 
      variable $v\in V$ of $I$ add $v$ to $T$ as an isolated root.
     \item Pick an edge $\{v,C\}\in\calE$ such that $v\in V(T)$ but there is no
	  $s$ such that $vC^s\in E(T)$ or $C^sv\in E(T)$.
	  (If no such edge exists, then output ``No'' and terminate.)
	  \label{step:Edges}
	\item  Add new node $C^t$ to $T$ together with the edge $vC^t$.
	  \label{step:AddConstraint}
      \item \label{step:ExploringChildren} Let $W$ be the set of all variables $w\ne v$ in the scope of $C$ such that $f(C)\oplus
	v\oplus w\in C$ (recall that $f(C)\oplus v\oplus w\in C$ is a shorthand for $f(C)\oplus v\oplus w\in R_C$). For each $w\in W$ do the following (see Figure \ref{fig:treealg}):\label{step:W}

        \begin{enumerate}
           \item If $w\notin V(T)$, then add $w$ to $T$ together with the edge
	     $C^t w$.\label{step:AddVariable}
           \item Else if $w$ has a parent of the form $C^s$ for some $s$, then
	     do nothing.\label{step:Nothing}
           \item Else if $v$ and $w$ belong to different trees in $T$ (i.e.~originate from different roots), then we have found an augmenting path.
                 Let $p=\walk(C^t), \walk(w)^{-1}$, output $f\oplus p$ and exit.\label{step:Augment}
           \item Else if $v$ and $w$ belong to the same tree in $T$, then we
	     have found a blossom. Form a new instance $I^b$ and new valid edge labeling
	     $f^b$ of $I^b$ by \emph{contracting}
                 this blossom. Solve this instance recursively, use the
		 resulting improved edge labeling for $I^b$ (if it exists) to compute an
		 improved valid edge labeling for $I$, and terminate. All details are
		 given in Sec.~\ref{section:contracting}\label{step:blossom}.
    \end{enumerate}
    \item Increase the timestamp $t$ by 1 and goto step 2.
  \end{enumerate}
\caption{Improving a given edge labeling}
\label{alg:Improve}
\end{algorithm*}

A vague answer to this objection is that we grow $T$ very carefully:
While the Edmonds' algorithm does not impose any
restrictions on the order in which the forest is grown, we require that all valid children $w\in W$ be added to $T$ simultaneously when exploring edge $\{v,C\}$ in step~\ref{step:ExploringChildren}.
Informally speaking, this will guarantee that forest $T$ does not have
``shortcuts'', a property that will be essential in the proofs. The possibility of having
shortcuts is something that is not present in graph matchings and is
one of
the properties of even $\Delta$-matroids responsible for the considerable length
of the correctness proofs.


In the following theorem, we collect all pieces we need to show that Algorithm~\ref{alg:Improve}
is correct and runs in polynomial time:

\begin{theorem}\label{thm:correctness}
  If $I$ is a CSP instance, $f$ is a valid edge labeling of $I$, and we run
  Algorithm~\ref{alg:Improve}, then the following is true:
\begin{enumerate}
  \item The mapping $f\oplus p$ from step~\ref{step:Augment} is a valid edge labeling
    of $I$ with fewer inconsistencies than $f$.\label{claimAugment}
  \item When contracting a blossom as described in
    Section~\ref{section:contracting} $I^b$ is an edge CSP instance with even $\Delta$-matroid constraints and $f^b$
    is a valid edge labeling to $I^b$.\label{claimContracting}
  \item The recursion in~\ref{step:blossom} will occur at most $O(|V|)$ many
    times.\label{claimRecursion}
  \item In step~\ref{step:blossom}, $f^b$ is optimal for $I^b$ if and
    only if $f$ is optimal for $I$. Moreover, given a valid edge labeling
    $g^b$ of $I^b$ with fewer inconsistent variables than $f^b$, we can in
    polynomial time output a valid edge labeling $g$ of $I$ with fewer
    inconsistent variables than $f$.\label{claimblossom} 
  \item If the algorithm answers ``No'' then $f$ is optimal.\label{claimNo}
\end{enumerate}
\end{theorem}
\begin{figure*}[t]
\begin{center}
\includegraphics[scale=1]{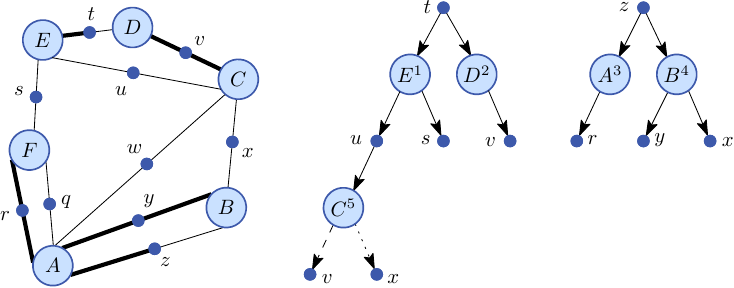}
\caption{A possible run of Algorithm~\ref{alg:Improve} on the instance $I'$ from Example
\ref{ex:instance} (with renamed constraint nodes) where the edge labeling $f$
is marked by thick (1) and thin (0) half-edges. We see that the algorithm finds
a blossom when it hits the variable $v$ the second time in the same tree.
However, had we first processed the transition $Cx$ (which we could have done), we would have found an augmenting path $p = \walk(C^5)
\walk(x)^{-1}$ (where $\walk(x)^{-1}$ ends in $z$).}
\label{fig:treealg}
\end{center}
\end{figure*}

\subsection{Contracting a blossom (step~\ref{step:blossom})}
\label{section:contracting}
We now elaborate step \ref{step:blossom} of Algorithm~\ref{alg:Improve}.
First, let us describe how to obtain the blossom~$b$. 
Let $\alpha\in V(T)\cup \calC(T)$ be the lowest common ancestor of nodes $v$ and $w$ in $T$.
Two cases are possible.

\begin{enumerate}
\item \label{contracting:case1}  $\alpha=r\in V(T)$.  Variable node $r$ must be inconsistent in $f$ because it has
outdegree two. We let $b=\walk(C^t),\walk(w)^{-1}$ in this case.
\item \label{contracting:case2} $\alpha=R^s\in \calC(T)$. 
Let $r$ be the child of $R^s$ in $T$ that is an ancestor of $v$. Replace edge labeling $f$ with $f\oplus \walk(r)$
(variable $r$ then becomes inconsistent).
Now define walk $b=p,q^{-1},r$
where $p$ is the walk from $r$ to $C^t$ in $T$ and $q$ is the walk from $R^s$ to
$w$ in $T$ (see Figure~\ref{fig:blossom}).
\end{enumerate}

\begin{lemma}[To be proved in Section~\ref{sec:AugmentContract}]\label{lemma:stem}
  Assume that Algorithm~\ref{alg:Improve} reaches step~\ref{step:blossom} and
  one of the cases described in the above paragraph occurs.
  Then:
  \begin{enumerate}
    \item in case~\ref{contracting:case2} the edge labeling $f\oplus \walk(r)$
      is valid, and
    \item in both cases the walk $b$ is an $f$-blossom (for the new edge
      labeling $f$, in case~\ref{contracting:case2}). 
      (Note that we have not formally defined $f$-blossoms yet; they require some machinery that
      will come later -- see Definition~\ref{def:Blossom}.)
  \end{enumerate}
\end{lemma}

To summarize, at this point we have a valid edge labeling $f$ of instance $I$
and an $f$-blossom 
$b=b_0 C_1 b_1 \dots C_k b_k$. Let us denote by $L$ the set of constraints in
the blossom, i.e.~$L=\{C_1,\dots,C_k\}$. 
\begin{figure*}[t]
\begin{center}
\includegraphics[scale=1]{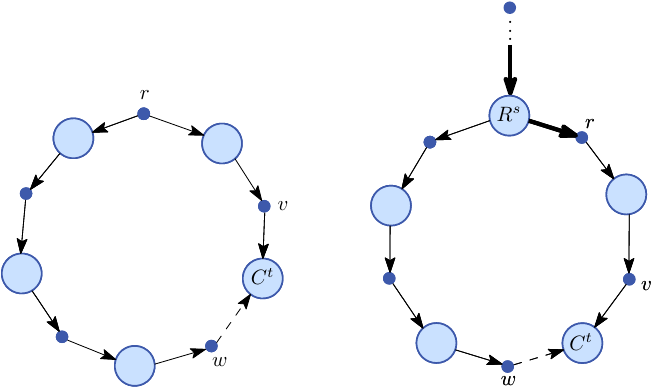}
\caption{The two cases of step~\ref{step:blossom}. On the left, $\alpha=r$
      is a variable, while on the right $\alpha=R^s$ is a constraint and the
      thick edges denote $p=\walk(r)$. The dashed edges are orientations of edges from $\calE$ that
  are not in the digraph $T$, but belong to the blossom.}\label{fig:blossom}
\end{center}
\end{figure*}

We construct a new instance $I^b$
and its valid edge labeling $f^b$ by \emph{contracting the blossom} $b$ as
follows: we take $I$, add one $|L|$-ary
constraint $N$ to $I$, delete the variables $b_1,\dots,b_{k}$, and add new
variables $\{v_C\colon C\in L\}$ (see Figure~\ref{fig:contraction}). The scope of $N$ is $\{v_C\colon C\in L\}$
and the $\Delta$-matroid of $N$ consists of exactly those maps $\alpha\in \{0,1\}^L$ that
send one $v_C$ to 1 and the rest to 0 (that is, $N$ is one of the perfect
matching $\Delta$-matroids from Example~\ref{ex:instance}).

In addition to all this, we replace each blossom constraint $D\in L$ by the constraint $D^b$ whose scope is
$\sigma\setminus \{b_1,\dots,b_k\}\cup \{v_D\}$ where $\sigma$ is the scope of
$D$. The constraint relation of $D^b$ consists of all maps $\beta$ for which there
exists $\alpha\in D$ such that $\alpha$ agrees with $\beta$ on
$\sigma\setminus\{b_1,\dots,b_k\}$ and one of the following occurs (see
Figure~\ref{fig:mulvisits}; note that $\sigma$ can contain more than two
elements of $\{b_1,\dots,b_k\}$ if $D$ appears in the blossom multiple times):
\begin{enumerate}
  \item $\beta(v_D)=0$ and $\alpha$ agrees with the original labeling $f(D)$ on all variables in $\{b_1,\dots,b_k\}\cap
    \sigma$, or
  \item\label{itm:flip-one} $\beta(v_D)=1$ and there is exactly one variable $z\in \{b_1,\dots,b_k\}\cap
    \sigma$ such that $\alpha(z)\neq f(D)(z)$.
\end{enumerate}

We claim that $D^b$ is an even $\Delta$-matroid. Indeed, let $Z_D$ be 
the relation on variables $\{b_1,\ldots,b_k\}\cap\sigma\cup\{v_D\}$
with the set of tuples
\[
 Z_D = \{\alpha\} \cup \{\alpha\oplus b \oplus v_D \:|\: b \in 
\{b_1,\ldots,b_k\}\cap \sigma \}
\]
where $\alpha$ is the tuple with $\alpha(b_i)=f(b_i)$ and $\alpha(v_D)=0$.
It is straightforward to verify that $Z_D$ is an even $\Delta$-matroid and that
$D^b$ is the composition of $D$ and $Z_D$ so, it follows from
Proposition~\ref{prop:composition} that each $D^b$ is an even $\Delta$-matroid. 
 
We define the edge labeling $f^b$ of $I^b$ as follows: for constraints
$A\notin\{C_1,\dots,C_k,N\}$ we set $f^b(A)=f(A)$. For each $C\in L$, we let
$f^b(C^b)(v)=f(C)(v)$ when $v\neq v_C$, and $f^b(C^b)(v_{C})=0$. Finally, we
let $f^b(N)(v_C)=1$ for $C=C_1$ and $f^b(N)(v_C)=0$ for all other $C$s. (The last
choice is arbitrary; initializing $f^b(N)$ with any other tuple in $N$ would
work as well.)

It is easy to check that $f^b$ is valid for $I^b$. Furthermore, $v_{C_1}$ is inconsistent in $f^b$ while
for each $C\in L\setminus\{C_1\}$ the variable $v_C$ is consistent.

\begin{observation}
  In the situation described above, the instance $I^b$ will have at most as
  many variables as $I$ and one constraint more than $I$. Edge labelings $f$ and 
  $f^b$ have the same number of inconsistent variables. 
\end{observation}

\begin{corollary}[Theorem~\ref{thm:correctness}(\ref{claimRecursion})]\label{cor:recursion}
  When given an instance $I$, Algorithm~\ref{alg:Improve} will recursively call
  itself $O(|V|)$ many times.
\end{corollary}
\begin{proof}
  Since $\calC$ and $V$ are partitions of $G_I$ and the degree of each $v\in V$
  is two, the number of edges of $G_I$ is $2|V|$. From the other
  side, the number of edges of $G_I$ is equal to the
  sum of arities of all constraints in $I$. Since we never consider constraints
  with empty scopes, the number of constraints of an instance is at most double
  the number of variables of the instance.

  Since each contraction adds one more constraint and never increases the
  number of variables, it follows that there cannot be a sequence of
  consecutive contractions longer than $2|V|$, which is $O(|V|)$.
\end{proof}

The following two lemmas, which we prove in Section~\ref{sec:proofs}, show why
the procedure works. In both lemmas, we let $(I,f)$ and $(I^b,f^b)$ denote the
instance and the valid edge labeling before and after the contraction,
respectively.
\begin{lemma}\label{lemma:ItoIprime}
  In the situation described above, if $f^b$ is optimal for $I^b$, then $f$ is optimal for $I$.
\end{lemma}
\begin{lemma}\label{lemma:IprimetoI}
  In the situation described above, if we are given a valid edge labeling $g^b$ of $I^b$ with
  fewer inconsistencies than $f^b$, then we can find in polynomial time a
  valid edge labeling $g$ of $I$ with fewer inconsistencies than $f$.
\end{lemma}

\begin{figure*}[t]
\begin{center}
\includegraphics[scale=1]{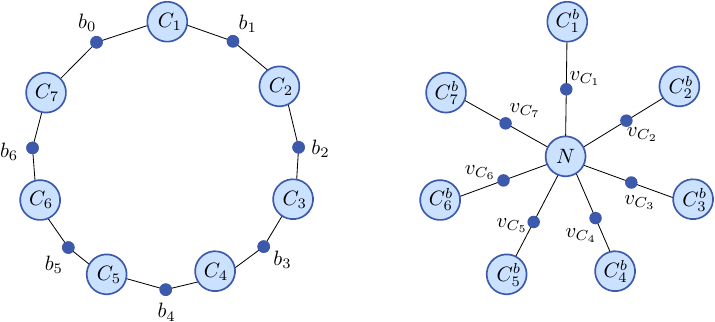}
\end{center}
\caption{A blossom (left) and a contracted blossom
  (right) in case when all constraints $C_1,\ldots,C_k$ are distinct.
If some constraints appear in the blossom multiple times then the number of
variables $v_{C_i}$
will be smaller than $k$ (see Figure~\ref{fig:mulvisits}).}\label{fig:contraction}
\end{figure*}
\begin{figure*}
\begin{center}
\includegraphics[width=\hsize]{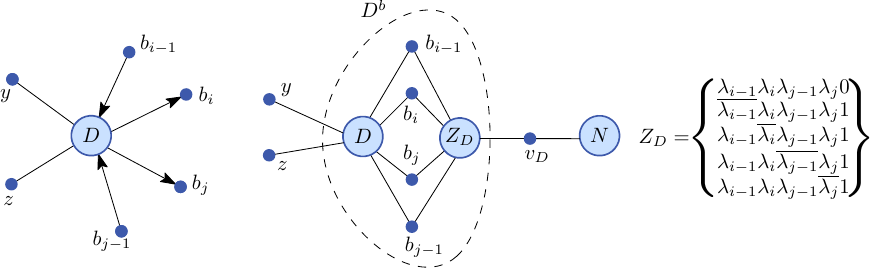}
\caption{Modification of a constraint node $D$ that appears in a blossom $b$
  twice, i.e.~when we have $b=\dots b_{i-1}Db_i \dots b_{j-1}Db_j \dots$ (and so $D=C_i=C_j$).
Variables $y$ and $z$ are not part of the walk. The construction of $D^b$ described in the text
can be alternatively viewed as composing $D$ with the $\Delta$-matroid $Z_D$ as shown in the figure.
Here $Z_D$ is an even $\Delta$-matroid with five tuples that depend on the values $\lambda_k = f(\{b_k,D\})$ and $\overline{\lambda_k} = 1- \lambda_k$.
}
\label{fig:mulvisits}
\end{center}
\end{figure*}

\subsection{Time complexity of Algorithm~\ref{alg:Improve}}
To see that Algorithm~\ref{alg:Improve} runs in time polynomial in the size of
$I$, consider first the case when step~\ref{step:blossom} does happen. In
this case, the algorithm runs in time polynomial in the size of $I$, since it
essentially just searches through the graph $G_I$. 

Moreover, from the description of contracting a blossom in
Section~\ref{section:contracting}, it is easy to see that one can compute $I^b$
and $f^b$ from $I$ and $f$ in polynomial time and that $I^b$ is not
significantly larger than $I$: $I^b$ has at most as many variables as $I$ and
the contracted blossom constraints $C^b$ are not larger than the
original constraints $C$. Finally, $I^b$ does have one brand new constraint $N$,
but $N$ contains only $O(|V|)$ many tuples. Therefore, we have $|I^b|\leq |I|+O(|V|)$ where $|V|$
does not change. By Corollary~\ref{cor:recursion}, there will be at most $O(|V|)$
contractions in total, so the size of the final instance $I^\star$ is at most
$|I|+O(|V|^2)$, which is easily polynomial in $|I|$.

All in all, Algorithm~\ref{alg:Improve} will give its answer in time polynomial
in $|I|$.
\section{Proofs}
\label{sec:proofs}
In this section, we flesh out detailed proofs of the statements
we gave above. In the whole section, $I$ will be an instance of a Boolean edge CSP whose
constraints are even $\Delta$-matroids.

In Sec.~\ref{sec:fwalks} we establish some properties of $f$-walks,
and show in particular that a valid edge labeling $f$ of $I$ is non-optimal if and only if
there exists an augmenting $f$-walk in $I$. In Sec.~\ref{sec:fDAGs} we introduce the notion of an \emph{$f$-DAG},
prove that the forest $T$ constructed during the algorithm is in fact an $f$-DAG,
and describe some tools for manipulating $f$-DAGs.
Then in Sec.~\ref{sec:AugmentContract} we analyze augmentation and contraction operations,
namely prove Theorem~\ref{thm:correctness}(\ref{claimAugment}) and Lemmas~\ref{lemma:stem},~\ref{lemma:ItoIprime},~\ref{lemma:IprimetoI}
(which imply Theorem~\ref{thm:correctness}(\ref{claimContracting}, \ref{claimblossom}). Finally,
in Sec.~\ref{sec:NoProof} we prove Theorem~\ref{thm:correctness}(\ref{claimNo}).

For edge labelings $f,g$, let $f \symdiff g\subseteq\calE$ be the set of edges
in $\calE$ on which $f$ and $g$ differ. 

\begin{observation}\label{obs:parity}
If $f$ and $g$ are valid edge labelings of instance $I$ then they have
the same number of inconsistencies modulo 2.
\end{observation}

\begin{proof}
We use induction on $|f \symdiff g|$. The base case $|f \symdiff g|=0$ is trivial.
For the induction step let us consider valid edge labelings $f,g$ with $|f \symdiff g|\ge 1$.
Pick an edge $\{v,C\}\in f \symdiff g$. 
By the property of even $\Delta$-matroids there exists another edge $\{w,C\}\in f \symdiff g$ with $w\ne v$
such that $f(C)\oplus v \oplus w\in C$. Thus, edge labeling $f^\star=f\oplus (vCw)$ is valid.
Clearly, $f$ and $f^\star$ have the same number of inconsistencies modulo 2.
By the induction hypothesis, the same holds for edge labelings $f^\star$ and $g$
(since $|f^\star \symdiff g|=|f \symdiff g|-2$). This proves the claim.
\end{proof}

\subsection{The properties of $f$-walks}\label{sec:fwalks}
Let us begin with some results on $f$-walks that will be of use later.
The following lemma is a (bit more technical) variant of the well known
property of labelings proven in~\cite[Theorem 3.6]{Dalmau2003}:
\begin{lemma}\label{lemma:Improve}
  Let  $f,g$ be valid edge labelings of $I$ such that $g$ has fewer
  inconsistencies than $f$, and $x$ be an inconsistent variable in $f$.
  Then there exists an augmenting $f$-walk that begins in a variable
  different from $x$. Moreover, such a walk can be computed in polynomial time
  given $I$, $f$, $g$, and $x$.
\end{lemma}
\begin{proof}
Our algorithm will proceed in two stages.
First, we repeatedly modify the edge labeling $g$ using the following procedure:
\begin{itemize}
\item[(1)] Pick a variable $v\in V$ which is consistent in $f$, but not in $g$. (If
  no such $v$ exists then go to the next paragraph).
By the choice of $v$, there exists a unique edge $\{v,C\}\in f \symdiff g$.
Pick a variable $w\ne v$ in the scope of $C$ such that $\{w,C\}\in f \symdiff g$ and $g(C)\oplus v\oplus w\in C$
(it exists since $C$ is an even $\Delta$-matroid). Replace $g$ with $g\oplus (vCw)$, then go to the beginning and repeat.
\end{itemize}
It can be seen that $g$ remains a valid edge labeling, and the number of inconsistencies in $g$ never increases.
Furthermore, each step decreases $|f \symdiff g|$ by $2$, so this procedure
must terminate after at most $O(|\calE|)=O(|V|)$ steps.

We now have valid edge labelings $f,g$ such that $f$ has more inconsistencies than $g$,
and variables consistent in $f$ are also consistent in $g$.
Since the parity of number of inconsistencies in $f$ and $g$ is the same,
$f$ 
has at least two more inconsistent variables than $g$;
one of them must be different from $x$.

In the second stage we will maintain an $f$-walk $p$ and the corresponding
 valid edge labeling $f^\star=f\oplus p$.
To initialize, pick a variable $r\in V\setminus\{x\}$
which is consistent in $g$ but not in $f$,
and set $p=r$ and $f^\star=f$. We then repeatedly apply the following step:
\begin{itemize}
  \item[(2)] Let $v$ be the endpoint of $p$. The variable $v$ is consistent in $g$ but not in $f^\star$,
so there is a unique edge $\{v,C\}\in f^\star \symdiff g$. 
Pick a variable $w\ne v$ in the scope of $C$ such that $\{w,C\}\in f^\star \symdiff g$ and $f^\star(C)\oplus v\oplus w\in C$
(it exists since $C$ is an even $\Delta$-matroid). Append $vCw$ to the end of $p$, and accordingly
replace $f^\star$ with $f^\star\oplus(vCw)$ (which is valid by the choice of $w$).
As a result of this update of $f^\star$, edges $\{v,C\}$ and $\{w,C\}$ disappear from $f^\star \symdiff g$.

If $w$ is inconsistent in $f$, then output $p$ (which is an augmenting $f$-walk) and terminate.
Otherwise $w$ is consistent in $f$ (and thus in $g$) but not in $f^\star$; in
this case, go to the beginning and repeat.
\end{itemize}
Each step decreases $|f^\star \symdiff g|$ by $2$, so this procedure
must terminate after at most $O(|\calE|)=O(|V|)$ steps. To see that $p$
is indeed a walk, observe that the starting node $r$ has exactly one incident
edge in the graph $(V\cup\calC,f^\star \symdiff g)$. Since this edge is immediately removed
from $f^\star \symdiff g$, we will never encounter the variable $r$ again during the procedure.
\end{proof}

\subsection{Invariants of Algorithm~\ref{alg:Improve}: $f$-DAGs}
\label{sec:fDAGs}
In this section we examine the properties of the forest $T$ as generated by
Algorithm~\ref{alg:Improve}. For future comfort, we will actually allow $T$ to be a
bit more general than what appears in Algorithm~\ref{alg:Improve} -- our $T$ can
be a  directed acyclic digraph (DAG):

\begin{definition}\label{def:DAG}
  Let $I$ be a Boolean edge CSP instance and $f$ a valid edge labeling of $I$.
  We will call a directed graph $T$ an \emph{$f$-DAG} if $T=(V(T)\cup\calC(T),E(T))$
where $V(T)\subseteq V$, $\calC(T)\subseteq\calC\times \en$, and
the following conditions hold:
\begin{enumerate}
  \item Edges of $E(T)$ have the form $vC^t$ or $C^tv$ where $\{v,C\}\in \calE$
    and $t\in\en$.\label{item:Form}
  \item\label{item:MostOnce} For each $\{v,C\}\in \calE$ there is at most one $t\in\en$ such that
    $vC^t$ or $C^tv$ appears in $E(T)$. Moreover, $vC^t$ and $C^tv$ are never
    both in $E(T)$.
  
  \item\label{item:Incoming} Each node $v\in V(T)$ has at most one incoming
    edge. (Note that by the previous properties, the node $v$ can have at most two incident edges in $T$.)
  \item \label{item:Timestamps}  
    Timestamps $t$ for nodes $C^t\in\calC(T)$ are all distinct (and thus 
    give a total order on $\calC(T))$. Moreover, this order can be extended to a
    total order $\prec$ on $V(T)\cup \calC(T)$ 
    such that $\alpha \prec \beta$ for each edge $\alpha\beta\in E(T)$. (So
    in particular the digraph $T$ is acyclic.)
  \item\label{item:SwitchTwo} If $T$ contains edges  $u C^t$ and one of $vC^t$ or $C^tv$, then $f(C)\oplus u \oplus v\in C$.
  \item\label{item:Shortcuts} (``No shortcuts'' property) If $T$ contains edges $u C^s$ and one of
    $vC^t$ or $C^tv$ where $s<t$, then $f(C)\oplus u\oplus v\notin C$.
\end{enumerate}
\end{definition}

From the definition of an $f$-DAG, we immediately obtain the following.
\begin{observation}\label{obs:sub-f-dag}
Any subgraph of an $f$-DAG is also an $f$-DAG. 
\end{observation}

 If $T$
is an $f$-DAG, then we denote by $f\oplus T$ the edge labeling we obtain from $f$
by flipping the value of any $f(\{v,C\})$ such that $vC^t\in E(T)$ or $C^tv\in
E(T)$ for some
timestamp $t$. We will need to show that $f\oplus T$ is a valid edge labeling for
nice enough $f$-DAGs $T$.


The following lemma shows the promised invariant property:
\begin{lemma}\label{lemma:AlgorithmDAG}
Let us consider the structure $T$ during the run of Algorithm~\ref{alg:Improve}
with the input $I$ and $f$. At any moment during the run, the forest $T$ is an $f$-DAG.

Moreover, if steps \ref{step:Augment} or \ref{step:blossom} are reached, then the
digraph $T^\star$ obtained from $T$ by removing all edges outgoing from $C^t$ and 
adding the edge $wC^t$ is also an $f$-DAG.
\end{lemma}
\begin{proof}
  Obviously, an empty $T$ is an $f$-DAG, as is the initial $T$ consisting of
  inconsistent variables and no edges. To verify that $T$ remains an $f$-DAG
  during the whole run of Algorithm~\ref{alg:Improve}, we need to make sure
  that neither adding $vC^t$ in step~\ref{step:AddConstraint}, nor adding $C^tw$ in
  step~\ref{step:AddVariable} violates the properties of $T$.
  Let us consider step~\ref{step:AddConstraint} first. By the choice of $v$ and 
  $C^t$, we immediately get that
  properties~(\ref{item:Form}), (\ref{item:MostOnce}), (\ref{item:Incoming}), and
  (\ref{item:Timestamps}) all hold even after we have added $vC^t$ to $T$ (we can
  order the nodes by the order in which they were added to $T$). Since
  there is only one edge incident with $C^t$, property~(\ref{item:SwitchTwo}) holds
  as well. Finally, the only way the ``no shortcuts'' property
  (i.e.~property~(\ref{item:Shortcuts})) could fail would be if there were some $u$ and
  $s<t$ such that $uC^s\in E(T)$ and $f(C)\oplus u\oplus v\in C$. But then, after
  the node $C^s$ got added to $T$, we should have computed the set $W$ of
  variables $w$ such that $f(C)\oplus u\oplus w\in C$
  (step~\ref{step:ExploringChildren}) and $v$
  should have been in $W\setminus V(T)$ at that time, i.e.~we should have added
  the edge $C^sv$ before, a contradiction. The analysis of step~\ref{step:AddVariable} is similar. 

  Assume now that Algorithm~\ref{alg:Improve} has reached one of
  steps~\ref{step:Augment} or \ref{step:blossom} and consider the DAG $T^\star$ that
  we get from $T$ by removing all edges of the form $C^tz$ and adding the edge $wC^t$. Note that the node $C^t$ is the only node
  with two incoming edges. The only three properties that could possibly be affected by going from $T$ to
  $T^\star$ are~(\ref{item:MostOnce}), (\ref{item:SwitchTwo}) and (\ref{item:Shortcuts}).
  Were~(\ref{item:MostOnce}) violated, we would have $C^sw\in E(T)$ already, and so
  step~\ref{step:Nothing} would be triggered instead of steps~\ref{step:Augment}
  or \ref{step:blossom}. For property~(\ref{item:SwitchTwo}), the only new pair of edges to consider is $vC^t$
  and $wC^t$ for which we have $f(C)\oplus v\oplus w\in C$. 
  Finally, if property~(\ref{item:Shortcuts}) became violated after adding the edge $wC^t$
  then there were a $u$ and $s<t$ such that $uC^s\in E(T)$ and $f(C)\oplus
  u\oplus w\in C$. 
  Node $C^s$ must have been added after $w$, or else we would have $C^sw\in
  E(T)$. Also, $w$ cannot have a parent of the form $C^k$ (otherwise step~\ref{step:Nothing} would be triggered for $w$ when expanding $C^t$).
  But then one of steps~\ref{step:blossom} or \ref{step:Augment} would
  be triggered at timestamp $s$ already when we tried to expand $C^s$, a
  contradiction.
\end{proof}

We will use the following two lemmas to prove that $f\oplus p$ is a valid
edge labeling of $I$ for various paths $p$ that appear in steps~\ref{step:Augment}
and~\ref{step:blossom}.

\begin{lemma}\label{lemma:shortening}
Let $T$ be an $f$-DAG, and $C^s$ be the constraint node in $\calC(T)$ with the smallest timestamp $s$.
Suppose that $C^s$ has exactly two incident edges, namely incoming edge $uC^s$
where $u$ does not have other incident edges besides $u C^s$ and another edge
$C^sv$ (see Figure \ref{fig:shortenfdag}). Let $f^\star=f \oplus (u C v)$ and let $T^\star$ be the DAG obtained from $T$ by removing nodes $u, C^s$ and the two edges incident to $C^s$.
Then $f^\star$ is a valid edge labeling of $I$ and $T^\star$ is an $f^\star$-DAG.
\end{lemma}
\begin{figure}[t]
  \begin{center}
    \includegraphics[scale=1]{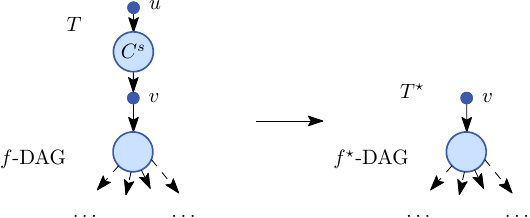}
  \end{center}
  \caption{An $f$-DAG $T$ on the left turns into $f^\star$-DAG $T^\star$ on the right; the setting from Lemma \ref{lemma:shortening}.}\label{fig:shortenfdag}
\end{figure}
\begin{proof}
  Since $T^\star$ is a subgraph of $T$, it immediately follows that $T^\star$
  satisfies the properties~(\ref{item:Form}), (\ref{item:MostOnce}), (\ref{item:Incoming}),
  and (\ref{item:Timestamps}) from the definition of an $f$-DAG all hold. 

  Let us show that $T^\star$ has property~(\ref{item:SwitchTwo}). 
  Consider a constraint node $C^t\in\calC(T^\star)$ with $t>s$ (nothing has changed for other
 constraint nodes in $\calC(T^\star)$), and suppose that
  $T^\star$ contains edges $xC^t$ and one of $yC^t$ or $C^ty$.  If $x=y$, the
  situation is trivial, so assume that $u,v,x,y$ are all distinct variables.
  We need to show that $f^\star(C)\oplus x\oplus y\in C$. The constraint $C$
  contains the tuples $f(C)\oplus u\oplus v$ and
  $f(C)\oplus x\oplus y$ (by condition~(\ref{item:SwitchTwo}) for $T$), but the no
  shortcuts property prohibits the tuples $f(C)\oplus u\oplus x$ and
  $f(C)\oplus u\oplus y$ from lying in $C$. Therefore, applying the even $\Delta$-matroid property on 
  $f(C)\oplus u\oplus v$ and
  $f(C)\oplus x\oplus y$ in the variable $u$ we get that $C$ must contain
  $f(C)\oplus u\oplus v\oplus x\oplus y$, so we have $f^\star(C)\oplus x\oplus y\in C$.

  Now let us prove that $T^\star$ and $f^\star$ have the ``no shortcuts'' property.
  Consider constraint nodes $C^k,C^{\ell}$ in $\calC(T^\star)$ with $s<k<\ell$
  (since nothing has changed for constraint nodes other than $C$), 
  and suppose that
  $T^\star$ contains edges $x C^k$ and one of $yC^{\ell}$ or $C^\ell y$, where again
  $u,v,x,y$ are all distinct variables.  We need to show that $f^\star(C)\oplus
  x\oplus y\notin C$, or equivalently that $f(C)\oplus u\oplus v\oplus x \oplus
  y\notin C$.

  Assume that it is not the case. Apply the even $\Delta$-matroid property to tuples
  $f(C)\oplus u\oplus v\oplus x \oplus y$ and $f(C)$ (which are both in $C$) in coordinate $v$.
  We get that either   $f(C)\oplus x \oplus y\in C$, or $f(C)\oplus u \oplus x\in C$, or $f(C)\oplus u \oplus y\in C$.
  This contradicts the ``no shortcuts'' property for the pair $(C^k,C^\ell)$, or
  $(C^s,C^k)$, or $(C^s,C^\ell)$, respectively, and we are done.
\end{proof}
\begin{corollary}\label{corollary:path}
  Let $I$ be an edge CSP instance and $f$ be a valid edge labeling.
  \begin{enumerate}
    \item Let $T$ be an $f$-DAG that consists of two directed paths
      $x_0C_1^{t_1}x_1\dots x_{k-1} C_k^{t_k}$ 
      and $y_0D_1^{s_1}\dots y_{\ell-1} D_\ell^{s_\ell}$ that are disjoint everywhere except 
  at the constraint $C_k^{t_k}=D_\ell^{s_\ell}$ (see Figure
  \ref{fig:twopaths}). Then $f\oplus T$ is a valid edge labeling of $I$.

  \item Let $T$ be an $f$-DAG that consists of a single directed path
    $x_0C_1^{t_0}x_1\dots x_{k-1} C_k^{t_k}x_{k}$.
  Then $f\oplus T$ is a valid edge labeling of $I$.
  \end{enumerate}
\end{corollary}

\begin{proof}
  We will prove only part (a); the proof of part (b) is completely analogous.
  We proceed by induction on $k+\ell$. If $k=\ell=1$, $T$ consists only
  of the two edges $x_0 C^t$ and $y_0C^t$ (where $C^t$ is an abbreviated name
  for $C_1^{t_1}=D_1^{s_1}$). Then the fact that $f\oplus (x_0 C y_0)$ is a
  valid edge labeling follows from the property~(\ref{item:SwitchTwo}) of $f$-DAGs.

  If we are now given an $f$-DAG $T$ of the above form, then we compare $t_1$
  and $s_1$. Since the situation is symmetric, we can assume without loss of
  generality that $s_1>t_1$. We then use Lemma~\ref{lemma:shortening} for
  $x_1C_1^{t_1}x_2$ (there is a $x_2$ since $t_k>s_1>t_1$),
  obtaining the $(f\oplus (x_1C_1x_2))$-DAG $T^\star$ that consists of two
  directed paths $x_2\dots x_k C^{t_k}$ and $y_1D_1^{s_1}\dots y_\ell
  D_\ell^{s_\ell}$. Since $T^\star$ is shorter than $T$, the induction
  hypothesis gets us that $f\oplus (x_1C_1x_2)\oplus T^\star=f\oplus T$ is a valid edge labeling.
\end{proof}
\begin{figure*}[t]
\begin{center}
\includegraphics[scale=1]{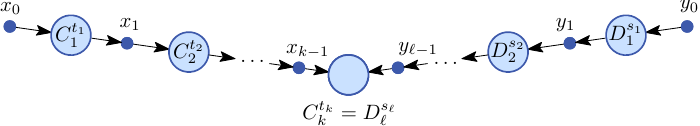}
\end{center} 
\caption{Two meeting paths from Corollary \ref{corollary:path}.}\label{fig:twopaths}
\end{figure*}

\begin{lemma}\label{lemma:shortening'}
  Let $T$ be an $f$-DAG, and $C^s$ be the constraint node in $\calC(T)$ with
  the smallest timestamp $s$.  Suppose that $C^s$ has exactly one incoming edge
  $u C^s$, and $u$ does not have other incident edges besides $u C^s$. Suppose
  also that $C^s$ has an outgoing edge $C^s v$.  Let $f^\star=f \oplus (u C v)$, and
  $T^\star$ be the DAG obtained from $T$ by removing the edge $uC^s$ together with $u$ and reversing the orientation of edge
  $C^sv$ (see Figure \ref{fig:shortenfdag'}).
  Then $f^\star$ is a valid edge labeling of $I$ and $T^\star$ is an $f^\star$-DAG.
\end{lemma}

\begin{figure*}[t]
\begin{center}
\includegraphics[scale=1]{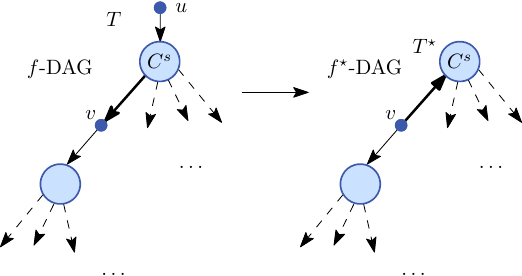}
\end{center}
\caption{An $f$-DAG $T$ turns into an $f^\star$-DAG $T^\star$ (see Lemma \ref{lemma:shortening'}).
} \label{fig:shortenfdag'}
\end{figure*}

\begin{proof}
  It is easy to verify that $T^\star$
  satisfies the properties~(\ref{item:Form}), (\ref{item:MostOnce}) and
  (\ref{item:Incoming}). To see property (\ref{item:Timestamps}), just take the linear
  order on nodes of $T$ and change the position of $v$ so that it is the new
  minimal element in this order ($v$ has no incoming edges in $T^\star$).

  Let us prove that property~(\ref{item:SwitchTwo}) of Definition~\ref{def:DAG} is preserved.
  First, consider constraint node $C^s$. Suppose that $T^\star$ contains one of $xC^s$ or $C^sx$ with $x\ne v$.
  We need to show that $f^\star (C)\oplus v\oplus x\in C$, or equivalently $f(C)\oplus u\oplus x\in C$
  (since $f^\star(C)\oplus v=f(C)\oplus (u\oplus v)\oplus v=f(C)\oplus u$).
  This claim holds by property~(\ref{item:SwitchTwo}) of Definition~\ref{def:DAG} for $T$.
  
  Now consider a constraint node $C^t\in\calC(T^\star)$ with $t>s$, and
  suppose that $T^\star$ contains edges $xC^t$ and one of $yC^t$ or $C^ty$.   We need to show that $f^\star(C)\oplus x\oplus y\in C$, or equivalently that $f(C)\oplus u\oplus v\oplus x \oplus y\in C$.
  For that we can simply repeat word-by-word the argument used in the proof of Lemma~\ref{lemma:shortening}.

  Now let us prove that the ``no shortcuts'' property is preserved. First,
  consider a constraint node $C^t$ in $\calC(T^\star)$ with $t>s$, and suppose that $T^\star$ contains one of $xC^t$ or $C^tx$.
  We need to show that $f^\star(C)\oplus v \oplus x\notin C$, or equivalently $f(C)\oplus u \oplus x\notin C$.
  This claim holds by the ``no shortcuts'' property for $T$.
  Now consider constraint nodes $C^k,C^\ell$
  in $\calC(T^\star)$ with $s<k<\ell$, and suppose that $T^\star$
  contains edges $x C^k$ and one of $yC^\ell$ or $C^\ell y$. Note that $u,v,x,y$ are all distinct variables. 
  We need to show that $f^\star(C)\oplus x\oplus y\notin C$, or equivalently that $f(C)\oplus u\oplus v\oplus x \oplus y\notin C$.
  For that we can simply repeat word-by-word the argument used to show the no
  shortcuts property in the proof of Lemma~\ref{lemma:shortening}.
\end{proof}

\subsection{Analysis of augmentations and contractions}
\label{sec:AugmentContract}
First, we prove the correctness of the augmentation operation.
\begin{proposition}[Theorem~\ref{thm:correctness}(\ref{claimAugment}) restated]
The mapping $f\oplus p$ from step~\ref{step:Augment} is a valid edge labeling
of $I$ with fewer inconsistencies than $f$.
\end{proposition}
\begin{proof}
  Let $T_1$ be the $f$-DAG constructed during the run of
  Algorithm~\ref{alg:Improve}; let $T_2$ be the DAG obtained from $T_1$ by
  adding the edge $wC^t$. By Lemma~\ref{lemma:AlgorithmDAG}, $T_2$ is an $f$-DAG.
  Let $T_3$ be the subgraph of $T_2$ induced by the nodes in $p$. It is easy
  to verify that $T_3$ consists of two directed paths that share their last
  node. Therefore, by Corollary~\ref{corollary:path}, we get that $f\oplus T_3=f\oplus
  p$ is a valid edge labeling of $I$.
\end{proof}

In the remainder of this section we show the correctness of the contraction operation
by proving
Lemmas~\ref{lemma:stem},~\ref{lemma:ItoIprime},~\ref{lemma:IprimetoI}.
Let us begin by giving a full definition of a blossom:
\begin{definition}\label{def:Blossom}
  Let $f$ be a valid edge labeling. An \emph{$f$-blossom} is any walk 
  $b=b_0 C_1 b_1 C_2 \dots C_{k} b_k$ with $b_0=b_k$  such that:
  \begin{enumerate}
    \item variable $b_0=b_k$ is inconsistent in $f$ while variables $b_1,\ldots,b_{k-1}$ are consistent, and
    \item there exists $\ell\in[1, k]$ and timestamps $t_1,\dots,t_k$ such that
      the DAG consisting of two directed paths $b_0C_1^{t_1}\dots b_{\ell-1}
      C_\ell^{t_\ell}$ 
      and $b_k C_k^{t_k} b_{k-1}\dots b_\ell C_\ell^{t_\ell}$ is an $f$-DAG.
  \end{enumerate}
\end{definition}

\begin{lemma}\label{observation:Intervals}
  Let $b$ be an $f$-blossom. Then $b_{[i,j]}$ is an $f$-walk
  (as per Definition~\ref{def:f-walk}) for 
any non-empty proper subinterval $[i,j]\subsetneqq[0,k]$.
\end{lemma}
\begin{proof}
  Let us denote the $f$-DAG from the definition of a blossom by $B$. By taking
  an appropriate subgraph of $B$ and applying Corollary~\ref{corollary:path}
we get that $f\oplus b_{[i,j]}$ is valid for any non-empty subinterval
$[i,j]\subsetneqq[0,k]$. Since the set of these intervals is downward closed,
$b_{[i,j]}$ is in fact an $f$-walk.
\end{proof}

\begin{lemma}[Lemma~\ref{lemma:stem} restated]
  Assume that Algorithm~\ref{alg:Improve} reaches step~\ref{step:blossom} and
  one of the cases described at the beginning of Section~\ref{section:contracting} occurs.
  Then:
  \begin{enumerate}
    \item in case~\ref{contracting:case2} the edge labeling $f\oplus \walk(r)$
      is valid, and
    \item in both cases the walk $b$ is an $f$-blossom (for the new edge
      labeling $f$, in case~\ref{contracting:case2}).
  \end{enumerate}
\end{lemma}

\begin{proof}
Let $T$ be the forest at the moment of contraction, $T^\dagger$ be the subgraph of $T$ containing only paths $\walk(C^t)$ and $\walk(w)$,
and $T^\star$ be the graph obtained from $T^\dagger$ by adding the edge $wC^t$.
By Lemma~\ref{lemma:AlgorithmDAG},
graph $T^\star$ is an $f$-DAG (any subgraph of an $f$-DAG is again an $f$-DAG;
  this is Observation~\ref{obs:sub-f-dag}).

If the lowest common ancestor of $w$ and $v$ in $T$ is a variable node $r\in
V(T)$ (i.e.~we have case~\ref{contracting:case1} from Section~\ref{section:contracting}), then the $f$-DAG $T^\star$ consists of two directed paths from $r$ to
the constraint $C$ and it is easy to verify that when we let $b$ to be one of
these paths followed by the other in reverse, we get a blossom.

Now consider case~\ref{contracting:case2}, i.e.~when the lowest common ancestor
of $w$ and $v$ in $T$ is a constraint node $R^s\in \calC(T)$.
Note that $T^\star$ has the unique source node $u$ (that does not have incoming edges),
and $u$ has an outgoing edge $uD^t$ where 
$D^t$ is the constraint node with the smallest timestamp in $T^\star$.
Let us repeat the following operation while $D^t\ne R^s$: Replace $f$ with $f\oplus (uD^tz)$ where
$z$ is the unique out-neighbor of $D^t$ in $T^\star$, and simultaneously modify $T^\star$ by
removing nodes $u, D^t$ and edges $uD^t,D^tz$. By Lemma~\ref{lemma:shortening}
$f$ remains a valid edge labeling throughout this process, and $T^\star$ remains an $f$-DAG (for the latest $f$).

We get to the point that the unique in-neighbor $u$ of $R^s$ is the source node of $T^\star$.
Replace $f$ with $f\oplus (uR^sr)$, and simultaneously modify $T^\star$ by
removing node $u$ together with the edge $uR^s$ and reversing the orientation of edge $R^sr$.
The new $f$ is again valid, and the new $T^\star$ is an $f$-DAG by Lemma~\ref{lemma:shortening'}.
This means that the resulting walk $b$ is an $f$-blossom for the new $f$.
\end{proof}

Finally, we prove two lemmas showing that if we contract a blossom $b$ in
instance $I$ to obtain
the instance $I^b$ and the edge labeling $f^b$, then $f$ is optimal for $I$ if and only if $f^b$ is optimal for $I^b$.

\begin{lemma}[Lemma~\ref{lemma:ItoIprime} restated]\label{lemma:ItoIprime2}
 In the situation described above, if $f^b$ is optimal for $I^b$, then $f$ is optimal for $I$.
 \end{lemma}

\begin{proof}
  Assume that $f$ is not optimal for $I$, so there exists a valid
  edge labeling $g$ with fewer inconsistencies than $f$. Then by Lemma~\ref{lemma:Improve}
  there exists an augmenting $f$-walk $p$ in $I$ that starts at some node other
  than $b_k$. Denote by $p^b$ the sequence obtained
  from $p$ by replacing each $C_i$ from the blossom by $C_i^b$. Observe that if $p$ does not
  contain the variables $b_1,\dots,b_k$, then $p$ is an $f$-walk 
  if and only if $p^b$ is an $f^b$-walk, so the only interesting case is when $p$
  enters the set $\{b_1,\dots,b_k\}$.

  We will proceed along $p$ and consider the first $i$ such that there is a blossom 
  constraint $D$ and an index $j$ for which  $p_{[0,i]}D b_{j}$ is an $f$-walk
  (i.e.~we can enter the blossom from $p$). 
  
  If $D=C_1$, then $p^b_{[0,i]}C_1^bv_{C_1}$ is an $f^b$-walk in
  $I^b$. To see that this is an $f^b$-walk, note that the labeling $f\oplus
  p_{[0,i]}C_1 b_{j}$ of $I$ agrees with $f$ on all edges of $b$ incident to
  $C_1$ except for
  $C_1b_j$, so it follows from part~\ref{itm:flip-one} of the definition of
  $C_1^b$ that the tuple $(f^b\oplus p^b_{[0,i]}C_1^b v_{C_1})(C_1^b)$ lies
  inside $C_1^b$. 

  If $D\neq C_1$, then similar arguments give us that
  $p^b_{[0,i]}D^bv_{C_D}Nv_{C_1}$ is an $f^b$-walk. In both
  cases, the $f^b$-walk found is augmenting (recall that the variable $v_{C_1}$ 
	is inconsistent in $f^b$). We found an augmentation of $f^b$, and so $f^b$ was not optimal.
\end{proof}

To show the other direction, we will first prove the following result. 
\begin{lemma}\label{lemma:f-DAG-plus-path}
	Let $q$ be an $f$-walk and $T$ an $f$-DAG such that
	$q\cap T\cap V=\emptyset$ and there is no proper prefix
  $q^\star$ of $q$ and no edge $vC^s$ or $C^sv$ of $T$ such that $q^\star Cv$ would
  be an $f$-walk. Then $T$ is a $(f\oplus q)$-DAG.
\end{lemma}
\begin{proof}
  We proceed by induction on the length of $q$. If $q$ has length 0, the claim
  is trivial. Otherwise, let $q=xCyq^{\dagger}$ for some $q^\dagger$. Note that $q^\dagger$ is
  trivially an $(f\oplus (xCy))$-walk. We verify that $T$ is an $(f\oplus (xCy))$-DAG, at which
  point it is straightforward to apply the induction hypothesis with $f\oplus xCy$ and $q^\dagger$ to 
  show that $T$ is an $(f\oplus q)$-DAG.
 
  We choose the timestamp $t$ to be smaller than any of the timestamps
  appearing in $T$ and construct the DAG $T^\dagger$ from $T$ by adding the
  nodes $x,y, C^{t}$ and edges $xC^t$ and $C^ty$. It is easy to see that $T^\dagger$ 
  is an $f$-DAG -- the only property
  that might possibly fail is the ``no shortcuts'' property. However, since the
  timestamp of $C^t$ is minimal, were the ``no shortcuts'' property violated, $T$
  would have to contain an edge of the form $vC^s$ or $C^sv$ such that 
  $f(C)\oplus x\oplus v\in C$. But in that case, we would have the $f$-walk
  $xCv$, contradicting our assumption on prefixes of $q$.
  
  It follows that $T^\dagger$ is an $f$-DAG and we can use 
  Lemma~\ref{lemma:shortening} with the
  constraint $C^t$ and edges $xC^t$ and $C^ty$ to show that $T$ is an
  $(f\oplus (xC^ty))$-DAG, concluding the proof.
\end{proof}

\begin{lemma}[Lemma~\ref{lemma:IprimetoI} restated]
  In the situation described above, if we are given a valid edge labeling $g^b$ of $I^b$ with
  fewer inconsistencies than $f^b$, then we can find in polynomial time a
  valid edge labeling $g$ of $I$ with fewer inconsistencies than $f$.
\end{lemma}

\begin{proof}
  Our overall strategy here is to take an inconsistency from the outside of the
  blossom $b$ and bring it into the blossom. We begin by showing how to get a
  valid edge labeling $f'$ for $I$ with an inconsistent variable just one edge
  away from $b$.

  Using Lemma~\ref{lemma:Improve}, we can use $g^b$ and $f^b$ to find in
  polynomial time an augmenting $f^b$-walk $p^b$ that does not begin at the
  inconsistent variable $v_{C_1}$. If $p^b$ does not contain any of the
  variables $v_{C_1},\dots,v_{C_k}$, then we can just output the walk $p$
  obtained from $p^b$ by replacing each $C_i^b$ by $C_i$ and be done. Assume
  now that some $v_C$ appears in $p^b$. We choose the $f^b$-walk $r^b$ so that
  $r^bC^bv_C$ is the shortest prefix of $p^b$ that ends with some blossom
  variable $v_C$.  By renaming all $C^b$s in $r^b$ to $C$s, we get the walk
  $r$. It is straightforward to verify that $r$ is an $f$-walk and that
  $rC_ib_i$ or $rC_ib_{i-1}$ is an $f$-walk for some $i\in[1,k]$. Let $q$ be
  the shortest prefix of $r$ such that one of $qC_ib_i$ or $qC_ib_{i-1}$ is an
  $f$-walk for some $i\in[1,k]$. 
  
  Recall that the blossom $b$ originates from an $f$-DAG $B$. The minimality of
  $q$ allows us to apply Lemma~\ref{lemma:f-DAG-plus-path} and obtain that $B$
  is also an $(f\oplus q)$-DAG. Let $f'=f\oplus q$  and let $x$ be the last
  variable in $q$. It is easy to see that $f'$ is a valid edge labeling with
  exactly as many inconsistent variables as $f$. Moreover $x$ is inconsistent
  in $f'$ and there is an index $i$ such that at least one of $xC_ib_i$ or
  $xC_ib_{i-1}$ is an $f'$-walk. We will now show how to improve $f'$.
  
  If the constraint $C_i$ appears only once in the blossom $b$, it is easy to
  verify (using Lemma~\ref{observation:Intervals}) that one of $xC_i b_{[i,k]}$
  or $xC_i {b_{[0,i-1]}}^{-1}$ is an augmenting $f'$-walk. However, since the
  constraint $C_i$ might appear in the blossom several times, we have to come
  up with a more elaborate scheme.  The blossom $b$ comes from
  an $f'$-DAG $B$ in which some node $C_\ell^{t_\ell}$ is the node with the
  maximal timestamp (for a suitable $\ell\in[1, k]$). Assume first that there
  is a $j\in[\ell,k]$ such that $xC_jb_{j}$ is an $f'$-walk. In that case, we
  take maximal such $j$ and consider the DAG $B'$ we get by adding the edge
  $C_j^{t_j}x$ to the subgraph of $B$ induced by the nodes
  $C_j^{t_j},b_j,C_{j+1}^{t_{j+1}},\dots, C_k^{t_k},b_k$, obtaining the
  directed path $b_k C_k^{t_k} b_{k-1} C_{k-1}^{t_{k-1}}\dots b_j C_j^{t_j} x$.

  It is routine to
  verify that $B'$ is an $f'$-DAG; the only thing that could possibly fail
  is the ``no shortcuts'' property involving $C_j^{t_j}$. However, $C_j^{t_j}$ has maximal
  timestamp in $B'$ and there is no $i>j$ such that $f'(C_j)\oplus x\oplus b_i \in
  C_j$. 
  
  Using Corollary~\ref{corollary:path}, we 
  get that $f'\oplus B'$ is a valid
  edge labeling which has fewer inconsistencies than $f'$, so we are done. 
  In a similar way, we can improve $f'$ when there exists a $j\in[1, \ell]$ such that
  $xC_j b_{j-1}$ is an $f'$-walk.

  \begin{figure}
  \begin{center}
  \includegraphics[scale=1]{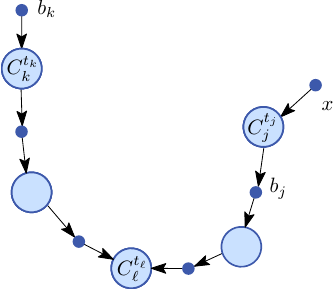}
  \caption{The $f'$-DAG $B'$ constructed using $xC_jb_{j}$ (where $j<\ell$).}
  \label{pic:lemma25}
  \end{center}
  \end{figure}

  If neither of the above cases occurs, then we take $j$ such that the
  timestamp $t_j$ is
  maximal and either $xC_{j}b_{j}$ or $xC_jb_{j-1}$ is an $f'$-walk. Without
  loss of generality, let $xC_jb_j$ be an $f'$-walk. Then $j< \ell$ and we
  consider the DAG $B'$ we get from the subgraph of $B$ induced by $C_j^{t_j},b_j,
  C_{j+1}^{t_{j+1}},\dots, C_k^{t_k},b_k$ by adding the edge $x C_j^{t_j}$ (see Figure~\ref{pic:lemma25}).
  As before, 
  the only way $B'$ cannot be an $f'$-DAG is if the ``no shortcuts'' property fails,
  but that is impossible: we chose $j$ so that $t_j$ is maximal, so an
  examination of the makeup of $B'$ shows that the only bad
  thing that could possibly happen is if there were an index $i\geq \ell$ such that
  $C_i=C_j$, we had in $B$ the edge $b_iC_i^{t_i}$, and $f'(C_i)\oplus b_i\oplus x\in C_i$.
  But then we would have the $f'$-walk $xC_ib_i$ for $i\geq \ell$ and the
  procedure from the previous paragraph would apply.
  Using Corollary~\ref{corollary:path}, we again see that $f'\oplus B'$ is a
  valid edge labeling with fewer inconsistencies than $f$. 
  
  It is easy to verify that finding $q$, calculating $f'=f\oplus q$, finding
  an appropriate $j$ and augmenting $f'$ can all be done in time
  polynomial in the size of the instance.
\end{proof}

\subsection{Proof of Theorem~\ref{thm:correctness}(\ref{claimNo})}
\label{sec:NoProof}
In this section we will prove that if the algorithm answers ``No'' then $f$ is an optimal edge labeling.
\begin{lemma}\label{lemma:reach}
  Suppose that Algorithm~\ref{alg:Improve} outputs ``No'' in step~\ref{step:Edges}, without ever visiting steps~\ref{step:Augment} and~\ref{step:blossom}.
  Then $f$ is optimal.
\end{lemma}

\begin{proof}
  Let $T$ be the forest upon termination. Our goal is to show that $T$
  describes all edges that can be reached from some inconsistent variable
  by an $f$-walk. In the paragraphs below, we make the meaning of ``describes''
  more precise.

  First of all, we define the set of edges present in $T$ (i.e. we forget the
  timestamps):
  \[
\overline E(T)=\{Cv\colon C^tv\in E(T)\mbox{ for some }t\}\cup 
\{vC\colon vC^t\in E(T)\mbox{ for some }t\}.
  \]
Inspecting Algorithm~\ref{alg:Improve}, one can check that $\overline E(T)$ has
the following properties:
\begin{enumerate}
\item \label{noproof:root} If $v$ is an inconsistent variable in $f$ and
  $\{v,C\}\in\calE$, then $vC\in\overline E(T)$.
\item \label{noproof:continue} If $Cv\in\overline E(T)$ and
  $\{v,D\}\in\calE,D\ne C$, then $vD\in\overline E(T)$.
\item \label{noproof:dir} If $vC\in\overline E(T)$, then $Cv\notin\overline E(T)$.
\item \label{noproof:transition} Suppose that $vC\in\overline E(T)$ and $f(C)\oplus v \oplus w\in C$ where $v,w$ are distinct nodes in the scope of constraint $C$. Then $Cw\in\overline E(T)$.
\end{enumerate}

  It is easy to see that for each $Cv\in\overline E(T)$ there is an $f$-walk
  that starts in an inconsistent variable and ends in $Cv$: Just take the
  directed path from a suitable root of $T$ to $C^tv$ in $T$ and apply
  Corollary~\ref{corollary:path}.

  Our goal in the rest of the proof is to show the converse -- if there is an
  $f$-walk that starts in an inconsistent variable and ends with the edge
  $Cv$ then $Cv\in \overline E(T)$. This will prove the Lemma: If $f$ is not 
  optimal then by Lemma~\ref{lemma:Improve} there is an augmenting $f$-walk
  that ends with an edge $Cv$ where $v$ is inconsistent. We thus should have
  $Cv\in\overline E(T)$. However, by
  property~\ref{noproof:root} above we have $vC\in\overline E(T)$ and thus (by
  property~\ref{noproof:dir}) $Cv\notin\overline E(T)$, a contradiction. 

  However, to be able to take a smallest counterexample, we will need to strengthen our
  statement, making it more local:  Call an $f$-walk 
  \emph{bad} if it starts at a variable node which is inconsistent in $f$, and
  contains (anywhere; not just at the end) an edge $Cv\notin \overline E(T)$; otherwise an $f$-walk is \emph{good}. We
  will show that bad $f$-walks do not exist, which in particular means that
  any nonzero length $f$-walk from an inconsistent variable needs to end with $Cv\in \overline
  E(T)$ and the argument from the previous paragraph applies.

  Assume for a contradiction that there exists a bad $f$-walk. Let $p$
  be a shortest bad walk. Write $p=p^\star(vCw)$ where $p^\star$ ends at $v$.
  By minimality of $p$, $p^\star$ is good and $Cw\notin\overline E(T)$.  Using
  properties~(\ref{noproof:root}) or (\ref{noproof:continue}),  we obtain that
  $vC\in\overline E(T)$ (and therefore $Cv\notin\overline E(T)$).

Let $q$ be the shortest prefix of $p^\star$ (also an $f$-walk) such that
the labeling $f\oplus q\oplus (vCw)$ is valid
(at least one such prefix exists, namely $q=p^\star$). The walk $q$ must be of
positive length (otherwise the precondition of
property~(\ref{noproof:transition}) would hold, and we would get $Cw\in\overline E(T)$, a contradiction).
Also, the last constraint node in $q$ must be $C$,
otherwise we could have taken a shorter prefix. Thus,
we can write $q=q^\star(xCy)$ where $q^\star$ ends at $x$. Note that, since $p$
is a walk, the variables $x,y,v,w$ are (pairwise) distinct.

We shall write $g=f\oplus q^\star$. Let us apply the even $\Delta$-matroid
  property to the tuples $g(C)\oplus x\oplus y\oplus v \oplus w$ and $g(C)$ 
(which are both in $C$) in coordinate $y$. We get that either $g(C)\oplus v \oplus w\in C$,
or $g(C)\oplus x \oplus v\in C$, or $g(C)\oplus x \oplus w\in C$.
In the first case we could have chosen $q^\star$ instead of $q$ -- a contradiction to the minimality of $q$.
In the other two cases $q^\star(xCu)$ is an $f$-walk for some $u\in\{v,w\}$.
But then from $Cu\notin \overline E(T)$ we get that $q^\star(xCu)$ is a bad walk  -- a contradiction to the minimality of $p$.
\end{proof}

\begin{corollary}[Theorem~\ref{thm:correctness}(\ref{claimNo})]\label{corollary:no}
  If Algorithm~\ref{alg:Improve} answers ``No'', then the edge labeling $f$ is
  optimal.
\end{corollary}
\begin{proof}
  Algorithm~\ref{alg:Improve} can answer ``No'' for two reasons: either the
  forest $T$ cannot be grown further and neither an augmenting path nor a
  blossom are found, or the algorithm finds a blossom $b$, contracts it and then
  concludes that $f^b$ is optimal for $I^b$. We proceed by induction on the
  number of contractions that have occurred during the run of the algorithm.
	
  The base case, when there were no contractions, follows from
  Lemma~\ref{lemma:reach}.  The induction step is an easy consequence of
  Lemma~\ref{lemma:ItoIprime} (also known as Lemma~\ref{lemma:ItoIprime2}):
  If we find $b$ and the algorithm answers ``No'' when run on $f^b$ and $I^b$
  then, by the induction hypothesis, $f^b$ is optimal for $I^b$, and by
  Lemma~\ref{lemma:ItoIprime} $f$ is optimal for $I$.
\end{proof}

\section{Extending our algorithm to efficiently co\-ver\-able $\Delta$-matroids} 
\label{sec:extend}

In this section we extend Algorithm~\ref{alg:Improve} from even $\Delta$-matroids to a wider
class of so-called efficiently coverable $\Delta$-matroids. The idea of the algorithm is similar to
what~\cite{feder-ford-matroids} previously did for $\mathcal C$-zebra
$\Delta$-matroids, but our method covers a larger class of $\Delta$-matroids.

Let us begin by giving a formal
definition of efficiently coverable $\Delta$-matroids.
\begin{definition}\label{def:efficiently-coverable} We say that a class of $\Delta$-matroids $\Gamma$ is
  \emph{efficiently coverable} if there is an algorithm that, given input $M\in
  \Gamma$ and $\alpha\in M$, lists in polynomial time a set $M_\alpha$ so that
  the system $\{M_\alpha\colon \alpha\in M\}$ 
  satisfies the conditions of Definition~\ref{def:coverable}.
\end{definition}

Before we go on, we would like to note that coverable $\Delta$-matroids are 
closed under gadgets, i.e. the ``supernodes'' shown in Figure~\ref{pic:example}. 
Taking gadgets is a common construction in the CSP world, so
being closed under gadgets makes coverable $\Delta$-matroids a very natural
class to study. Taking gadgets is equivalent to repeated composition of 
$\Delta$-matroids (see Proposition~\ref{prop:composition}).

\begin{definition}
Given $M\subset \{0,1\}^U$ and $N\subset \{0,1\}^V$ with $U,V$ disjoint sets of
variables, we define the \emph{direct product} of $M$ and $N$ as 
\[
  M\times N=\{(\alpha,\beta)\colon \alpha\in M,\beta\in N\}\subset
  \{0,1\}^{U\cup V}.
\]
If $w_1,w_2\in U$ are distinct variables of $M$, then the $\Delta$-matroid obtained from
$M$ by identifying $w_1$ and $w_2$ is
\begin{align*}
  M_{w_1=w_2}&=\{\beta_{\upharpoonright U\setminus\{w_1,w_2\}} \colon 
  \beta\in M,\,\beta(w_1)=\beta(w_2)\} \\
  &\subset \{0,1\}^{U\setminus\{w_1,w_2\}}
\end{align*}
\end{definition}
Since both above operations are special cases of $\Delta$-matroid compositions,
by Proposition~\ref{prop:composition}, (even) $\Delta$-matroids
are closed under direct product and identifying variables.

If a $\Delta$-matroid $P$ is obtained from some $\Delta$-matroids $M_1,\dots,M_k$ by a sequence
of direct products and identifying variables, we say that $P$ is
\emph{gadget-constructed}
from $M_1,\dots,M_k$ (a gadget is an edge CSP instance with some variables
present in only one constraint -- these are the ``output variables'').

\begin{theorem}
  The class of coverable $\Delta$-matroids is closed under:
  \begin{enumerate}
    \item Direct products,
    \item identifying pairs of variables, and
    \item compositions,
    \item gadget constructions.
  \end{enumerate}
\end{theorem}
\begin{proof}
\begin{enumerate}
  \item Let $M\subset\{0,1\}^U$ and $N\subset\{0,1\}^V$ be two coverable
      $\Delta$-matroids. We claim that if $(\alpha,\beta)$ and $(\gamma,\delta)$
      are even-neighbors in $M\times N$, then either $\alpha=\gamma$
      and $\beta$ is an even-neighbor of $\delta$ in $N$, or $\beta=\delta$
      and $\alpha$ is an
      even-neighbor of $\gamma$ in $M$. This is
      straightforward to verify: Without loss of generality let us assume that
      $u\in U$ is a variable of
      $M$ such that $(\alpha,\beta)\oplus u\not\in M\times N$, and let $v$ be
      the variable such that $(\alpha,\beta)\oplus u\oplus v=(\gamma,\delta)$. 
      Since we are dealing with a direct product, we must have
      $\alpha\oplus u\not\in M$ and in order for $(\alpha,\beta)\oplus u\oplus
      v$ to lie in $M\times N$, we must have $\alpha \oplus u\oplus
      v\in M$. But then $\delta=\beta$ and $\alpha\oplus u\oplus v=\gamma$ is an even-neighbor of
      $\alpha$.

      Let $\alpha\in M$, $\beta\in N$ and let $M_\alpha$, $N_\beta$ be the even
      $\Delta$-matroids from Definition~\ref{def:coverable} for $M$ and $N$. 
      From the above paragraph, it follows by induction that whenever
      $(\gamma,\delta)$ is reachable from
      $(\alpha,\beta)$, then $(\gamma,\delta)\in M_\alpha\times N_\beta$. Since
      each $M_\alpha\times N_\beta$ is an even $\Delta$-matroid, the direct product
      $M\times N$ satisfies the first two parts of Definition~\ref{def:coverable}.

      It remains to show that if we can reach
      $(\gamma,\delta)\in M\times N$ from  $(\alpha, \beta)\in M\times N$ 
      and $(\gamma,\delta)\oplus u\oplus v\in M_\alpha\times
      N_\beta\setminus M\times N$, then $(\gamma,\delta)\oplus u,(\gamma,\delta)\oplus v\in M\times N$. 
      By the first paragraph of this proof, we can reach
      $\gamma$ from $\alpha$ in $M$ and $\delta$ from $\beta$ in $N$.
      Moreover, both $u$ and $v$ must lie in the same set $U$ or $V$, for otherwise
      we would have that 
      $(\gamma\oplus v,\delta\oplus u)$ or $(\gamma\oplus u,\delta\oplus v)$
      lies in $M_\alpha\times
      N_\beta$, a contradiction with
      $M_\alpha$ being an even $\Delta$-matroid. So let (again without loss of
      generality) $u,v\in V$. Then $\delta\oplus u\oplus v\in N_\beta\setminus
      N$. Since $N$ is coverable and $\delta$ is reachable from $\beta$, 
      we get
      $\delta\oplus v,\delta\oplus u\in N$, giving us $(\gamma,\delta)\oplus
      u,(\gamma,\delta)\oplus v\in M\times N$ and we are done.
      
    \item Let $M\subset \{0,1\}^U$ be coverable and $w_1\neq w_2$ be two
      variables.

      Similarly to the previous item, the key part of the proof is to show that the
      relation of being reachable survives identifying $w_1$ and $w_2$: More
      precisely, take $\alpha,\gamma\in M_{w_1=w_2}$ and $\beta\in M$ such that
      $\beta(w_1)=\beta(w_2)$ and $\alpha=\beta_{\upharpoonright
      U\setminus\{w_1,w_2\}}$ (i.e. $\beta$ witnesses  $\alpha\in M_{w_1=w_2}$).
      Assume that we can reach $\gamma$ from $\alpha$.
      Then we can reach from $\beta$ a tuple $\delta\in M$ such that
      $\delta(w_1)=\delta(w_2)$ 
      and $\gamma=\delta_{\upharpoonright U\setminus\{w_1,w_2\}}$.

      Since we can proceed by induction, it is enough to prove this claim in
      the case when $\alpha,\gamma$ are even-neighbors. So assume that there
      exist variables $u$ and $v$ such that  $\gamma=\alpha\oplus u\oplus v$
      and $\alpha\oplus u\not\in M_{w_1=w_2}$. From the latter, it follows that
      $\beta\oplus u,\beta\oplus u\oplus w_1\oplus w_2 \not\in M$.  Knowing all
      this, we see that if $\beta\oplus u \oplus v\in M$, the tuple
      $\beta\oplus u\oplus v$ is an even-neighbor of $\beta$ and we are done.

      Suppose, to the contrary, that $\beta\oplus u\oplus v\not\in M$. Let $\delta$ be the
      tuple of $M$ witnessing $\gamma\in M_{w_1=w_2}$. Since $\beta\oplus
      u\oplus v\not\in M$, we get $\delta=\beta\oplus u\oplus v\oplus w_1\oplus
      w_2$. Since $\beta\oplus u,\beta\oplus u\oplus
      v\not\in M$, the $\Delta$-matroid property applied on $\beta$ and
      $\delta$ in the variable $u$ gives us (without loss of generality) that
      $\beta\oplus u\oplus w_1\in M$. But then $\beta$ is an even-neighbor of
      $\beta\oplus u\oplus w_1$ in $M$, which is an even neighbor (via the variable
      $w_2$ -- recall that $\beta\oplus u\oplus w_1\oplus w_2\not\in M$) of
      $\delta$ and so we can reach $\delta$ from $\beta$, proving the claim.

      Assume now that $M$ is coverable. We want to show that the sets
      $(M_\beta)_{w_1=w_2}$ where $\beta$ ranges over $M$ cover $M_{w_1=w_2}$. Choose $\alpha\in M_{w_1=w_2}$ and let
      $\beta\in M$ be the witness for $\alpha\in M_{w_1=w_2}$. We claim that the even
      $\Delta$-matroid $(M_\beta)_{w_1=w_2}$ contains all members of
      $M_{w_1=w_2}$ that can be reached from $\alpha$. Indeed, whenever
      $\gamma$ can be reached from $\alpha$, some $\delta\in M$ that 
      witnesses $\gamma\in M_{w_1=w_2}$ can be reached from $\beta$, so
      $\delta\in M_\beta$ and $\gamma\in(M_\beta)_{w_1=w_2}$.

      To finish the proof, take $\beta\in M$
      witnessing $\alpha\in M_{w_1=w_2}$ and 
      $\gamma\in M_{w_1=w_2}$ that is reachable from
      $\alpha$ and satisfies $\gamma\oplus u\oplus v\in (M_\beta)_{w_1=w_2}\setminus
      M_{w_1=w_2}$ for a suitable pair of variables $u,v$. Take a $\delta\in M$ that 
      witnesses $\gamma\in M_{w_1=w_2}$ and is reachable from $\beta$ (we have
      shown above that such a $\delta$ exists). Since $\gamma\oplus u\oplus v\in (M_\beta)_{w_1=w_2}\setminus
      M_{w_1=w_2}$, we know that neither
      $\delta\oplus u\oplus v$ nor
      $\delta\oplus u\oplus v\oplus w_1\oplus w_2$ lies in $M$, but at least one
      of these two tuples lies in $M_\beta$. If $\delta\oplus u\oplus v\in
      M_\beta$, we just use coverability of $M$ to get $\delta\oplus
      u,\delta\oplus v\in M$, which translates to $\gamma\oplus u,\gamma\oplus
      v\in M_{w_1=w_2}$. If this is not the case, we know that 
      $\delta, \delta\oplus u\oplus v\oplus w_1\oplus w_2\in M_\beta$ and
      $\delta\oplus u\oplus v\not\in M_\beta$. We show that in this situation
      we have $\gamma\oplus u\in M_{w_1=w_2}$; the proof of $\gamma\oplus v\in
      M_{w_1=w_2}$ is analogous. 
      
      Using the
      even $\Delta$-matroid property of $M_\beta$ on $\delta$ and $\delta\oplus
      u\oplus v\oplus w_1\oplus w_2$ in the variable $u$, we get that without loss of generality
      $\delta\oplus u\oplus w_1\in M_\beta$ (recall that $\delta\oplus u\oplus v\not\in
      M_\beta$). If $\delta\oplus u\oplus w_1\not\in M$, we can directly
      use coverability of $M$ on $\delta$ to get that $\delta\oplus u\in
      M$, resulting in $\gamma\oplus u\in M_{w_1=w_2}$. If, on the other hand, $\delta\oplus u\oplus w_1\in M$
      and $\delta\oplus u\not\in M$, then $\delta\oplus u\oplus w_1$ is
      reachable from $\beta$, so we can use coverability of $M$ on 
      $\delta\oplus u\oplus w_1\in M$ and      $\delta\oplus u\oplus v\oplus
      w_1\oplus w_2\in M_\beta\setminus M$ to get   $\delta\oplus u\oplus w_1\oplus w_2\in M$, which
      again results in $\gamma\oplus u\in M_{w_1=w_2}$, finishing the proof.
    \item Since a composition of two $\Delta$-matroids is just a
	    direct product followed by a series of identifying
	    variables, it follows from previous points that
	    coverable $\Delta$-matroids are closed under
	    compositions.
    \item This follows from first two points as any gadget construction 
      is equivalent to a sequence of products followed by identifying variables.
\end{enumerate}
\end{proof}

Returning to edge CSP, the main notions from the even
$\Delta$-matroid case translate to the efficiently coverable $\Delta$-matroid
case easily. The definitions of valid, optimal, and non-optimal edge
labeling may remain intact for coverable $\Delta$-matroids, but we need to
adjust our definition of a walk, which will now be allowed to end in a
constraint.

\begin{definition}[Walk for general $\Delta$-matroids] \label{def:generalwalk} 
  A \emph{walk} $q$ of length $k$ or $k+1/2$ in the instance $I$ is a sequence
$q_0 C_1 q_1 C_2 \dots C_{k} q_k$ or $q_0 C_1 q_1 C_2 \dots C_{k+1}$, respectively,
where the variables $q_{i-1},q_i$ lie in the
scope of the constraint $C_i$, and each edge $\{v,C\}\in\calE$ is traversed at most once:
 $vC$ and $Cv$ occur in $q$ at most once, and they do not occur simultaneously.
\end{definition}

Given  an edge labeling $f$ and a walk $q$, we define the edge labeling
$f\oplus q$ in the same way as before (see eq.~\eqref{eq:foplusq:def}).
We also extend the definitions of an $f$-walk and an
augmenting $f$-walk for a valid edge labeling $f$: A walk $q$ is an $f$-walk if
$f\oplus q^\star$ is a valid edge
labeling whenever $q^\star=q$ or $q^\star$ is a prefix of $q$ that ends at a variable.
An $f$-walk is called augmenting if: (1) it starts at a variable inconsistent in $f$,
 (2) it ends either at a different inconsistent variable or in a constraint, and
 (3) all variables inside of $q$ (i.e.. not endpoints) are consistent in $f$. Note that if $f$ is a valid edge labeling for which there is an augmenting $f$-walk, then $f$ is non-optimal
(since $f\oplus q$ is a valid edge labeling with 1 or 2 fewer inconsistent variables).

The main result of this section is tractability of efficiently coverable
$\Delta$-matroids. 

\begin{theorem}[Theorem \ref{thm:extension2} restated]
Given an edge CSP instance $I$ with efficiently coverable $\Delta$-matroid constraints,
an optimal edge labeling $f$ of $I$ can be found in time polynomial in $|I|$.
\end{theorem}

The rough intuition of the algorithm for improving coverable $\Delta$-matroid
edge CSP instances is the following. When dealing with
general $\Delta$-matroids, augmenting $f$-walks may also end in a
constraint -- let us say that $I$ has the augmenting $f$-walk $q$ that ends in
a constraint $C$. In that case, the parity of $f(D)$ and $(f\oplus p)(D)$ is the same for all
$D\neq C$. If we guess the correct $C$ (in fact, we will try all options) and flip its
parity, we can, under reasonable conditions, find this augmentation via the algorithm for even
$\Delta$-matroids.

Not all $\Delta$-matroids $M$ are coverable.
However, we will show below how to efficiently cover many previously considered classes of
$\Delta$-matroids. These would be co-independent
\cite{feder-delta-matroids-fanout}, compact \cite{Istrate97lookingfor}, local
\cite{Dalmau2003}, linear~\cite{Geelen2003377} and binary \cite{Geelen2003377,Dalmau2003}
$\Delta$-matroids (note that in the case of the last two our representation of
the $\Delta$-matroid is different from~\cite{Geelen2003377}).

\begin{proposition}\label{prop:covers} The classes of co-independent, local,
  compact, linear and binary $\Delta$-matroids are efficiently coverable.
\end{proposition}

One part of Proposition~\ref{prop:covers} that is easy to prove is efficient
coverability of linear $\Delta$-matroids: Every linear $\Delta$-matroid is even
because the tuples in the $\Delta$-matroid correspond to regular skew-symmetric
matrices and every skew-symmetric matrix of odd size is singular
(see~\cite{Geelen2003377} for the definition and details). Thus our basic
algorithm already solves edge CSP with linear $\Delta$-matroid constraints
(should we represent our constraints by lists of tuples and not matrices).

For the rest of the proof of this proposition as well as (some of) the
definitions, we refer the reader to Appendix~\ref{app:covers}. 

\subsection{The algorithm}

The following lemma is a straightforward generalization of the result
given in Lemma~\ref{lemma:Improve}.

\begin{lemma}\label{lemma:generalaugpath}
  Let  $f,g$ be valid edge labelings of instance $I$ (with general $\Delta$-matroid constraints) such that $g$ has fewer
  inconsistencies than $f$. Then we can, given $f$ and $g$, compute in
  polynomial time an
  augmenting $f$-walk $p$ (possibly ending in a constraint, in the sense of
  Definition \ref{def:generalwalk}). 
\end{lemma}
\begin{proof} 
  We proceed in two stages like in the proof of Lemma~\ref{lemma:Improve}:
  First we modify $g$ so that any variable consistent in $f$ is consistent in
  $g$, then we look for the augmenting $f$-walk in $f\symdiff g$. The only
  difference over Lemma~\ref{lemma:Improve} is that our $g$-walks and $f$-walks
  can now end in a constraint as well as in a variable.
  
First, we repeatedly modify the edge labeling $g$ using the following procedure:
\begin{itemize}
  \item[(1)] Pick a variable $v\in V$ which is consistent in $f$, but not in $g$. (If
  no such $v$ exists then go to the next paragraph).
By the choice of $v$, there exists a unique edge $\{v,C\}\in f \symdiff g$. If $g(C) \oplus v \in C$, replace $g$ with $g \oplus vC$, then go to the beginning and repeat.
Otherwise, pick variable $w\ne v$ in the scope of $C$ such that $\{w,C\}\in f \symdiff g$ and $g(C)\oplus v\oplus w\in C$
(it exists since $C$ is a $\Delta$-matroid and $g(C) \oplus v \not\in C$). Replace $g$ with $g\oplus (vCw)$ and then also go to the beginning and repeat.
\end{itemize}
It can be seen that $g$ remains a valid edge labeling, and the number of inconsistencies in $g$ never increases.
Furthermore, each step decreases $|f \symdiff g|$, so this procedure
must terminate after at most $O(|\calE|)=O(|V|)$ steps.

We now have valid edge labelings $f,g$ such that $f$ has more inconsistencies than $g$,
and variables consistent in $f$ are also consistent in $g$.
In the second stage we will maintain an $f$-walk $p$ and the corresponding
 valid edge labeling $f^\star=f\oplus p$.
To initialize, pick a variable $r\in V$
which is consistent in $g$ but not in $f$,
and set $p=r$ and $f^\star=f$. We then repeatedly apply the following step:
\begin{itemize}
  \item[2.] Let $v$ be the endpoint of $p$. The variable $v$ is consistent in $g$ but not in $f^\star$,
so there must exist a unique edge $\{v,C\}\in f^\star \symdiff g$. If
    $f^\star(C) \oplus v \in C$, then output $pC$ (an augmenting $f$-walk). Otherwise, pick variable $w\ne v$ in the scope of $C$ such that $\{w,C\}\in f^\star \symdiff g$ and $f^\star(C)\oplus v\oplus w\in C$
(it exists since $C$ is a $\Delta$-matroid and $f^\star(C)\oplus v \not\in C$). Append $vCw$ to the end of $p$, and accordingly
replace $f^\star$ with $f^\star\oplus(vCw)$ (which is valid by the choice of $w$).
As a result of this update of $f^\star$, edges $\{v,C\}$ and $\{w,C\}$
    disappear from $f^\star \symdiff g$.

If $w$ is inconsistent in $f$, then output $p$ (which is an augmenting $f$-walk) and terminate.
Otherwise $w$ is consistent in $f$ (and thus in $g$) but not in $f^\star$; in
this case, go to the beginning and repeat.
\end{itemize}
It is easy to verify that the $p$ being produced is an $f$-walk. Also, each step decreases $|f^\star \symdiff g|$ by $2$, so this procedure
must terminate after at most $O(|\calE|)=O(|V|)$ steps and just like in the
case of even $\Delta$-matroids, the only way to terminate is to find an
augmentation. 
\end{proof}

\begin{definition} Let $f$ be a valid edge labeling of instance $I$ with
  coverable $\Delta$-matroid constraints. For a constraint $C \in {\mathcal C}$
  and a $\Delta$-matroid $C' \subseteq C$, we will denote by $I(f, C, C')$ the
  instance obtained from $I$ by replacing the constraint relation of $C$ by
  $C'$ and the constraint relation of each $D
  \in {\mathcal{C}}\setminus\{C\}$ by the even $\Delta$-matroid $D_{f(D)}$ (that
  comes from the covering). 
\end{definition}

Observe that $f$ induces a valid edge labeling for $I(f,C,C)$. Moreover, if we
choose $\alpha\in C$, then $I(f, C, \{\alpha\})$ is an edge CSP instance with even
$\Delta$-matroid constraints and hence we can find its optimal edge labeling by
Algorithm \ref{alg:Improve} in polynomial time.

\begin{lemma}\label{lemma:extensionworks} Let $f$ be a non-optimal valid edge labeling of instance $I$ with coverable $\Delta$-matroid constraints. Then there exist $C \in {\mathcal C}$ and $\alpha \in C$ such that the optimal edge labeling for $I(f, C, \{\alpha\})$ has fewer inconsistencies than $f$.
\end{lemma} 

\begin{proof} If $f$ is non-optimal for $I$, then by Lemma
  \ref{lemma:generalaugpath} there exists an augmenting $f$-walk $q$ in $I$.
  Take $q$ such that no proper prefix of $q$ is augmenting (i.e. we cannot end
  early in a constraint). Let $C$ be the last constraint in the walk and
  let $\alpha = (f \oplus q)(C)$. 
  
  We claim that $f \oplus q$ is also a valid edge labeling for
  the instance $I(f, C,\{\alpha\})$.  Since we choose $\alpha$ so that
  $(f\oplus q)(C)=\alpha$, we only need to consider constraints different from
  $C$.  Assume that $p$ is the shortest prefix of $q$ such
  that $(f\oplus p)(D)$ is not reachable from $D_{f(D)}$ for some $D\neq C$ (if
  there is no such thing, then $(f\oplus q)(D)\in D_{f(D)}$ for all $D\neq C$).
  We let $p=p^\star xDy$. Since $(f\oplus p^\star)(D)$ is reachable from
  $f(D)$, but $(f\oplus p^\star)(D)\oplus x\oplus y$ is not, we must have
  $(f\oplus p^\star)(D)\oplus x\in D$. But then $p^\star xD$ is an augmenting
  $f$-walk in $I$ that is shorter than $p$, a contradiction with the choice of
  $p$.
\end{proof}  

\begin{lemma}\label{lemma:getAP} Let $f$ be a valid assignment for the instance $I$ with coverable $\Delta$-matroid constraints and let $C \in {\mathcal C}$ and $\alpha \in C$ be such that there exists a valid edge labeling $g$ for the instance $I(f, C, \{\alpha\})$ with fewer inconsistencies than $f$. Then there exists an augmenting $f$-walk for $I$ and it can be computed in polynomial time given $g$.
\end{lemma}
\begin{proof} 
We begin by noticing that both $f$ and $g$ are valid edge labelings for the
  instance $I(f, C, C)$. Since $g$ has fewer inconsistencies than $f$,
by Lemma \ref{lemma:generalaugpath} we can compute an $f$-walk $q$ which
is augmenting in $I(f, C, C)$. It is easy to examine $q$ and check if some
proper prefix of $q$ is an augmenting $f$-walk for $I$ (ending in a
constraint). If that happens we are done, so let us assume that
this is not the case. We will show that then $q$ itself must be an augmenting $f$-walk for $I$.

First assume that every prefix of $q$ with integral length is an $f$-walk in
$I$. Then either $q$ is of integral length and we are
done ($q$ is its own prefix), or $q$ ends in a constraint. If it is the latter,
$q$ must end in $C$, since that is the only constraint of $I(f,C,C)$ that is
not forced to be an even $\Delta$-matroid. But the constraint relation $C$ is
the same for both $I$ and $I(f,C,C)$, so flipping the last edge of $q$ is
allowed in $I$.

Let now $p$ be the shortest prefix of $q$ with integral
length which is not an $f$-walk in $I$.
We can write $p = p^\star xDy$ for suitable $x,y,D$. The constraint relation of $D$ must be different in $I$ and
$I(f,C,C)$, so $D\neq C$. By the choice of $p$, for any prefix $r$ of $p^\star$ of integral length
we have $(f\oplus r)(D)\in D$ and moreover the tuple $(f\oplus r)(D)$ is reachable
  from $f(D)$. (If not,
  take the shortest counterexample $r$. Obviously, $r=r^\star uDv$ for some
  variables $u,v$ and a suitable $r^\star$.
Since $(f\oplus r^\star)(D)\in D$ is reachable from
$f(D)$ and $(f\oplus r^\star)(D)\oplus u\oplus v$ is not, we get $(f\oplus
r^\star)(D)\oplus u\in D$ and $r^\star uD$ is augmenting in $I$, which is a
  contradiction.) This holds also for $r=p^\star$, so $(f\oplus p^\star)(D)$ is
  reachable from $f(D)$.

To finish the proof, let $\beta^\star
= (f \oplus p^\star)(D)$ and $\beta=(f\oplus p)(D)$. 
We showed that $\beta^\star\in D$ is reachable from $f(D)$.
Also, $\beta^\star\oplus x\oplus y=\beta\in D_{f(D)}\setminus D$. Then by the
  definition of coverable $\Delta$-matroids we have
$\beta^\star\oplus x\in D$. Thus $p^\star xD$ is an augmenting $f$-walk in $I$
and we are done.

%

It is easy to see that all steps of the proof can be made algorithmic.
\end{proof}

Now the algorithm is very simple to describe. Set some valid edge labeling $f$
and repeat the following procedure.  For all pairs $(C, \alpha)$ with $\alpha
\in C$ and $C \in {\mathcal C}$, call Algorithm \ref{alg:Improve} on the instance
$I(f,C,\{\alpha\})$ (computing the instance $I(f,C,\{\alpha\})$ can be done in polynomial time
because all constraints of $I$ come from an efficiently coverable class). If for some
$(C, \alpha)$ we obtained an edge labeling of $I(f,C,\{\alpha\})$ with fewer inconsistencies, use Lemma
\ref{lemma:getAP} to get an augmenting $f$-walk for $I$.  Otherwise, we have
proved that the original $f$ was optimal.

The algorithm is correct due to Lemma \ref{lemma:extensionworks}. The running
time is
polynomial because there are at most $|I|$ pairs $(C, \alpha)$ such that $\alpha
\in C$ and at most $|I|$ inconsistencies in the initial edge labeling, so the
(polynomial) Algorithm \ref{alg:Improve} gets called at most $|I|^2$ times.

\subsection{Even-zebras are coverable (but not vice versa)}
The paper~\cite{feder-ford-matroids} introduces several classes of zebra
$\Delta$-matroids. For simplicity, we will consider only one of them:
$\calC$-zebras.
\begin{definition}
  Let $\calC$ be a subclass of even $\Delta$-matroids.
  A $\Delta$-matroid $M$ is a $\calC$-zebra if for every $\alpha\in M$ there
  exists an even $\Delta$-matroid $M_\alpha$ in $\calC$ that contains all tuples in $M$ of
  the same parity as $\alpha$ and such that for every $\beta\in M$ and every
  $u,v\in V$ such that $\beta\oplus u\oplus v\in M_\alpha\setminus M$ we have
  $\beta\oplus v,\beta\oplus u\in M$.
\end{definition}

In~\cite{feder-ford-matroids}, the authors show a result very much similar to
Theorem~\ref{thm:extension2}, but for $\calC$-zebras: In our language, the
result states that if one can find optimal labelings for $\EdgeCSP(\calC)$ in
polynomial time, then the same is true for the edge CSP with $\calC$-zebra
constraints. 
In the rest of this section, we show that coverable
$\Delta$-matroids properly contain the class of $\calC$-zebras with $\calC$
equal to all even $\Delta$-matroids (this is the largest $\calC$ allowed in the
definition of $\calC$-zebras) -- we will call this class \emph{even-zebras} for short.
We need to assume, just like in~\cite{feder-ford-matroids}, that we are given the zebra 
representations on input.

\begin{observation}
  Let $M$ be an even-zebra. Then $M$ is coverable.
\end{observation}
\begin{proof}
  Given $\alpha\in M$, we can easily verify that the $\Delta$-matroids $M_\alpha$
  satisfy all conditions of the definition of coverable $\Delta$-matroids:
  Everything reachable from $\alpha$ has the same parity as $\alpha$ and the
  last condition from the definition of even-zebras is identical to coverability.
\end{proof}

Moreover, it turns out that the inclusion is proper: There exists a
$\Delta$-matroid that is coverable, but is not an even-zebra. 

Let us take $M=\{(0,0,0)$, $(1,1,0)$, $(1,0,1)$, $(0,1,1)$, $(1,1,1)\}$ and consider $N=M\times
M$. It is easy to verify that $M$ is a $\Delta$-matroid that is an even-zebra with
the sets $M_\alpha$ equal to $\{(0,0,0)$, $(1,1,0)$, $(1,0,1)$, $(0,1,1)\}$ and
$\{(1,1,1)\}$, respectively, and thus $M$ is 
coverable. 

Since coverable $\Delta$-matroids are closed under direct products,
$N$ is also coverable.
However, $N$ is not an even-zebra: Assume that there exists a set $N_\alpha$
that contains all tuples of $N$ of odd parity and satisfies the zebra
condition. Then  the two tuples $(1,1,1,0,0,0)$ and
$(1,1,0,1,1,1)$ of $N$ belong to $N_\alpha$. Since $N_\alpha$ is an even
$\Delta$-matroid, switching in the third coordinate yields that $N$
contains the tuple $(1,1,0,1,0,0)$ (this is without loss of generality; the other
possibilities are all symmetric). This tuple is not a member of $N$, yet we got
it from $(1,1,1,0,0,0)\in N$ by switching the third and fourth coordinate. So in
order for the zebra property to hold, we need $(1,1,1,1,0,0)\in N$, a contradiction.

The above example also shows that even-zebras, unlike coverable
$\Delta$-matroids, are not closed under direct
products.

\section*{APPENDIX}
\appendix
\section{Non matching realizable even $\Delta$-matroid}
\label{sec:appendix}

Here we prove Proposition \ref{prop:notmatchrel} which says that not every even
$\Delta$-matroid of arity six is matching realizable. We do it by first showing
that matching realizable even $\Delta$-matroids satisfy certain decomposition
property and then we exhibit an even $\Delta$-matroid of arity six which does
not possess this property and thus is not matching realizable.

\begin{lemma}\label{lemma:pairs} Let $M$ be a matching realizable even $\Delta$-matroid and let $f,g \in M$. Then $f \symdiff g$ can be partitioned into pairs of variables $P_1, \dots P_k$ such that $f \oplus P_i \in M$ and $g \oplus P_i \in M$ for every $i = 1 \dots k$.
\end{lemma}

\begin{proof} Fix a graph $G=(N, E)$ that realizes $M$ and let  $V = \{v_1,
  \dots, v_n\} \subset N$ be the nodes corresponding to variables of $M$.
  Let $E_f$ and $E_g$ be the edge sets from matchings that correspond to tuples
  $f$ and $g$. Now consider the graph $G' = (N, E_f \symdiff E_g)$ (symmetric
  difference of matchings). Since both $E_f$ and $E_g$ cover each node of $N
  \setminus V$, the degree of all such nodes in $G'$ will be zero or two.
  Similarly, the degrees of nodes in $\left(V \setminus (f \symdiff
  g)\right)$ are either zero or two leaving $f \symdiff g$ as the set of
  nodes of odd degree, namely of degree one. Thus $G'$ is a union of induced
  cycles and paths, where the paths pair up the nodes in $f \symdiff g$. Let
  us use this pairing as $P_1$, \dots, $P_k$.

Each such path is a subset of $E$ and induces an alternating path with respect to both $E_f$ and $E_g$. After altering the matchings accordingly, we obtain new matchings that witness $f \oplus P_i \in M$ and $g \oplus P_i \in M$ for every $i$.
\end{proof}

\begin{lemma} There is an even $\Delta$-matroid of arity 6 which does not have the property from Lemma \ref{lemma:pairs}.
\end{lemma}
\begin{proof}
Let us consider the set $M$ with the following tuples:

\begin{center}
\begin{tabular}{cccc}
  000000&100100&011011&111111\\
	&011000&100111\\
	&001100&110011\\
	&001010&110101\\
	&000101&111010\\
	&001001&001111\\
	&010001&101101\\
	&100010&101011\\
	&      &111100\\
\end{tabular}
\end{center}
With enough patience or with computer aid one can verify that this is indeed an
even $\Delta$-matroid. However, there is no pairing satisfying the conclusion of Lemma~\ref{lemma:pairs} for
tuples $f = 000000$, and $g = 111111$. In fact the set of
pairs $P$ for which both $f \oplus P \in M$ and $g \oplus P \in M$ is $\{v_1,
v_4\}$, $\{v_2, v_3\}$, $\{v_3, v_4\}$, $\{v_3, v_5\}$, $\{v_4, v_6\}$ (see the
first five lines in the middle of the table above) but no three of these form a partition on $\{v_1, \dots, v_6\}$.
\end{proof}

\section{Classes of $\Delta$-matroids that are efficiently coverable}
\label{app:covers}

As we promised, here we will show that all classes of $\Delta$-matroids that were previously known to be tractable are efficiently coverable.

\subsection{Co-independent $\Delta$-matroids}

\begin{definition} A $\Delta$-matroid $M$ is \emph{co-independent} if whenever $\alpha \not\in M$, then $\alpha \oplus u \in M$ for every~$u$ in the scope of $M$. 
\end{definition}

Let $V$ be the set of variables of $M$. In this case we choose $M_\alpha$ to be
the $\Delta$-matroid that contains all members of $\{0,1\}^V$ of the same
parity as $\alpha$. This trivially satisfies the first two conditions in the
definition of a $\Delta$-matroid. To see the third condition, observe that
whenever $\gamma\in M_\alpha\setminus M$, the co-independence of $M$ gives us
that $\gamma\oplus u\in M$ for every $u\in V$, so we are done. 

Moreover, each set $M_\alpha$ is roughly as large as $M$ itself: A
straightforward double counting argument gives us that $M\geq 2^{|V|-1}$, so
listing $M_\alpha$ can be done in time linear in $|M|$. 

\subsection{Compact $\Delta$-matroids}

We present the definition of compact $\Delta$-matroids in an alternative form compared to \cite{Istrate97lookingfor}. 

\begin{definition} Function $F \colon \{0,1\}^V \to \{0,\dots, |V|\}$ is called a \emph{generalized counting function} (gc-function) if
\begin{enumerate}
\item for each $\alpha \in \{0,1\}^V$ and $v \in V$ we have $F(\alpha \oplus v) = F(\alpha) \pm 1$ and;
\item if $F(\alpha) > F(\beta)$ for some $\alpha, \beta \in \{0,1\}^V$, then there exist $u,v \in \alpha \symdiff \beta$ such that $F(\alpha \oplus u) = F(\alpha)-1$ and $F(\beta \oplus v) = F(\beta) + 1$
\end{enumerate}
\end{definition}

An example of such function is the function which simply counts the number of ones in a tuple.

\begin{definition} We say that a $S \subset \{0, 1, \dots n\}$ is \emph{$2$-gap free} if whenever $x \not\in S$ and $\min S < x < \max S$, then $x+1, x-1 \in S$.
 A set of tuples $M$ is \emph{compact-like} if $\alpha \in M$ if and only if $F(\alpha) \in S$ for some gc-function $F$ and a $2$-gap free subset $S$ of $\{0, 1, \dots |V|\}$. 
\end{definition}

The difference to the presentation in \cite{Istrate97lookingfor} is that they
give an explicit set of possible gc-functions (without using the term
gc-function). However, we decided for more brevity and omit the description of
the set. 

\begin{lemma}\label{lemma:compact} Each compact-like set of tuples $M$ is a $\Delta$-matroid.
\end{lemma} 
\begin{proof}
  Let the gc-function $F$ and the 2-gap free set $S$ witness that $M$ is
  compact-like. Take $\alpha, \beta \in M$ and $u \in \alpha \symdiff \beta$. If 
  $F(\alpha \oplus u) \in S$, then $\alpha \oplus u \in M$ and we are done. Thus we
  have $F(\alpha \oplus u) \neq F(\beta)$. We need to find a $v\in
  \alpha\symdiff\beta$ such that
  $F(\alpha\oplus u\oplus v)\in S$.

Let us assume $F(\alpha \oplus u) > F(\beta)$. Since $F$ is a gc-function we can 
  find $v \in (\alpha \oplus u) \symdiff \beta$ (note that
$u \neq v$)  such that $F(\alpha \oplus u
\oplus v) = F(\alpha \oplus u) -1$. Now we have either
$F(\alpha)=F(\alpha\oplus u\oplus v)\in S$, or $F(\alpha)>F(\alpha\oplus u)>F(\alpha\oplus
u\oplus v)\geq F(\beta)$, which again means $F(\alpha\oplus u\oplus v)\in S$ because
$S$ does not have 2-gaps.

The case when $F(\alpha \oplus u) < F(\beta)$ is handled analogously.
\end{proof}

It turns out that any practical class of compact-like $\Delta$-matroids is
efficiently coverable:
\begin{lemma}
  Assume $\mathcal M$ is a class of compact-like $\Delta$-matroids where the
  description of each $M\in
  \mathcal M$ includes a set $S_M$ (given by a list of elements) and a function $F_M$ witnessing that $M$ is
  compact-like and there is a polynomial $p$
  such that the time to compute $F_M(\alpha)$ is at most $p(|M|)$. Then $\mathcal M$
  is efficiently coverable.
\end{lemma}

\begin{proof}
  Given $M\in \mathcal M$ and $\alpha\in M$, we let $M_\alpha$ be the
  compact-like even $\Delta$-matroid given by the function $F_M$ and the set
  $U=[\min S_M,\max S_M]\cap \{F_M(\alpha)+2k\colon k\in \zet\}$. 
  It is an easy observation that $\alpha,\beta\in\{0,1\}^V$ have the same parity if
  and only if $F_M(\alpha)$ and $F_M(\beta)$ have the same parity, so all
  members of $M_\alpha$ have the same parity. In particular $M_\alpha$ contains
  all $\beta\in M$ of the same parity as $\alpha$. Moreover, the set $U$ is 2-gap
  free, so $M_\alpha$ is an even $\Delta$-matroid.

  Let now $\gamma\in M_\alpha\setminus M$. Then $F_M(\gamma)\not \in S_M$.
  Since $S_M$ is 2-gap free and $F_M(\gamma)$ is not equal to $\min S_M$, nor $\max S_M$, it follows that both $F_M(\gamma)+1$ and
  $F_M(\gamma)-1$ lie in $S_M$. Therefore, $\gamma\oplus v\in M$ for any $v\in
  V$ by the first property of gc-functions.

  It remains to show how to construct $M_\alpha$ in polynomial time. We begin
  by adding to $M_\alpha$ all tuples of $M$ of the same parity as $\alpha$.
  Then we go through all tuples $\beta\in M$ of parity different from $\alpha$
  and for each such $\beta$ we calculate $F_M(\beta\oplus v)$ for all $v\in V$. If
  $\min S<F(\beta\oplus v)<\max S$, we add $\beta\oplus v$ to $M_\alpha$. 
  By the argument in the previous paragraph, this procedure will eventually
  find and add to $M_\alpha$ all tuples $\gamma$ such that 
  $F_M(\gamma)\in U\setminus S_M$.
\end{proof}

\subsection{Local and binary $\Delta$-matroids}

We will avoid giving the definitions of local and binary $\Delta$-matroids.
Instead, we will rely on a result from \cite{Dalmau2003} saying that both of
these classes avoid a certain substructure. This will be enough to show that
both binary and local $\Delta$-matroids are efficiently coverable.
\begin{definition}
  Let $M,N$ be two $\Delta$-matroids where $M\subset \{0,1\}^V$. We say that $M$ contains $N$ as a
  \emph{minor} if we can get $N$ from $M$ by a sequence of the following
  operations: Choose $c\in\{0,1\}$ and $v\in V$ and take the $\Delta$-matroid we obtain by fixing the
    value at $v$ to $c$ and deleting $v$:
    \begin{align*}
      M_{v=c}=&\{\beta\in \{0,1\}^{V\setminus\{v\}}\colon 
      \exists \alpha \in M,\,
	\alpha(v)=c
	\wedge \forall u\neq v,\, \alpha(u)=\beta(u)\}.
    \end{align*}
\end{definition}

\begin{definition}
The interference $\Delta$-matroid is the ternary $\Delta$-matroid given by the
tuples $\{(0,0,0)$, $(1,1,0)$, $(1,0,1)$, $(0,1,1)$, $(1,1,1)\}.$
We say that a
$\Delta$-matroid $M$ is \emph{interference free} if it does not contain any minor 
isomorphic (via renaming variables or flipping the values 0 and 1 of some
variables) to the interference $\Delta$-matroid.
\end{definition}

\begin{lemma} \label{lemma:OddNeighbors} If $M$ is an interference-free
  $\Delta$-matroid and $\alpha,\beta\in M$ are
such that $|\alpha\symdiff \beta|$ is odd, then we can find $v\in \alpha
\symdiff \beta$ so that $\alpha\oplus v\in M$.
\end{lemma}

\begin{proof}
Let us take $\beta'\in M$ so that $\alpha\symdiff\beta'\subset \alpha\symdiff
\beta$ and $|\alpha\symdiff \beta'|$ is
odd and minimal possible. If $|\alpha\symdiff \beta'|=1$, we are done. Assume thus that
$|\alpha\symdiff \beta'|=2k+3$ for some $k\in \en_0$. Applying the $\Delta$-matroid property on $\alpha$ and $\beta'$
(with $\alpha$ being the tuple changed) $k$ many times, we get a set of $2k$
variables $U\subset \alpha\symdiff \beta'$ such that $\alpha\oplus U\in M$ (since $\beta'$ is at
minimal odd distance from $\alpha$, in each step we need to switch exactly two variables of
$\alpha$). 
  
  Let the three
variables in $\alpha\symdiff \beta'\setminus U$ be $x$, $y$, and $z$ and
  consider the $\Delta$-matroid $P$ on
$x,y,z$ we get from $M$ by fixing the values of all $v\not\in\{x,y,z\}$ to
those of $\alpha\oplus U$ and deleting these variables afterward. Moreover, we
switch 0s and 1s so that the triple corresponding to
$(\alpha(x),\alpha(y),\alpha(z))$ is
$(0,0,0)$. 
We claim that $P$ is the interference $\Delta$-matroid: It contains the triple $(0,0,0)$
(because of $\alpha\oplus U$) and $(1,1,1)$ (as witnessed by $\beta'$) and does not
contain any of the triples $(1,0,0)$, $(0,1,0)$, or $(0,0,1)$ (for then $\beta'$
would not be at minimal odd distance from $\alpha$). Applying the $\Delta$-matroid
property on $(1,1,1)$ and $(0,0,0)$ in each of the three variables then
necessarily gives us the tuples $(0,1,1)$, $(1,0,1)$, and $(1,1,0)\in P$.
\end{proof}

\begin{corollary} Let $M$ be an interference-free $\Delta$-matroid. 
If $M$ contains at least one even tuple then the set
$\Even(M)$ of all even tuples of $M$ forms a $\Delta$-matroid. 
The same holds for $\Odd(M)$ the set of all odd tuples of $M$. In particular,
$M$ is efficiently coverable by the even $\Delta$-matroids $\Even(M)$ and $\Odd(M)$.
\end{corollary}

\begin{proof}
  We show only that $\Even(M)$ is a $\Delta$-matroid; the case of $\Odd(M)$ is
  analogous and the covering result immediately follows.

  Take $\alpha,\beta\in \Even(M)$ and let $v$ be a variable $v$ such that $\alpha(v)\neq
  \beta(v)$. We want
  $u\neq v$ so that $\alpha(u)\neq \beta(u)$ and $\alpha\oplus u \oplus v\in M$. Apply the
  $\Delta$-matroid property of $M$ to $\alpha$ and $\beta$, changing the tuple
  $\alpha$. If we get $\alpha\oplus v\oplus u\in M$ for some $u$, we are done, so
  let us assume that we get $\alpha\oplus v\in M$ instead. But then we
  recover as follows: The tuples $\alpha\oplus v$ and $\beta$ have different parity, so
  by Lemma~\ref{lemma:OddNeighbors} there exists a variable $u$ so that $(\alpha\oplus v)(u)\neq
  \beta(u)$ (i.e. $u\in \alpha\symdiff \beta\setminus\{v\}$) and $\alpha\oplus
  v\oplus u\in M$.
\end{proof}

It is mentioned in \cite{Dalmau2003} (Section 4) that the interference
$\Delta$-matroid is among the forbidden minors for both local and binary
(minors B1 and L2) $\Delta$-matroids. Thus both of those classes are
efficiently coverable.

\begin{acks}
Most of this work was done while the authors were with IST
  Austria. This work was supported by \grantsponsor{ERC}{European Research
  Council}{https://erc.europa.eu/} under the
European Unions Seventh Framework Programme (FP7/2007-2013)/ERC grant
  agreement no \grantnum{ERC}{616160}.
\end{acks}

\bibliographystyle{ACM-Reference-Format}
\bibliography{citations}
\end{document}